\documentclass[10pt, doublecolumn]{IEEEtran}
\usepackage{epsfig,latexsym}
\usepackage{float}
\usepackage{indentfirst}
\usepackage{amsmath}
\usepackage{bm}
\usepackage{amssymb}
\usepackage{times}

\usepackage{algorithm}
\usepackage[noend]{algpseudocode}
\usepackage{subfigure}
\usepackage{psfrag}
\usepackage{hyperref}
\usepackage{cite}
\usepackage{lastpage}
\usepackage{fancyhdr}
\usepackage{color}
 \usepackage{amsthm}
\usepackage{bigints}
\usepackage{booktabs}
\usepackage{multirow}
\usepackage{siunitx}

\newtheorem{Lemma}{Lemma}
\newtheorem{Corollary}{Corollary}
\newtheorem{lemma}[Lemma]{$\mathbf{Lemma}$}

\newtheorem{corollary}[Corollary]{$\mathbf{Corollary}$}

\newcounter{problem}
\newcounter{save@equation}
\newcounter{save@problem}
\makeatletter
\newenvironment{problem}
{\setcounter{problem}{\value{save@problem}}%
  \setcounter{save@equation}{\value{equation}}%
  \let\c@equation\c@problem
  \subequations
}
{\endsubequations
  \setcounter{save@problem}{\value{equation}}%
  \setcounter{equation}{\value{save@equation}}%
}

\begin{document}
\title{  \huge{Design of Downlink Hybrid NOMA Transmission  }}

\author{\author{ Zhiguo Ding, \IEEEmembership{Fellow, IEEE},  Robert Schober, \IEEEmembership{Fellow, IEEE},  and H. Vincent Poor, \IEEEmembership{Life Fellow, IEEE}   \thanks{ 
  
\vspace{-2em}

  Z. Ding is with  Khalifa University, Abu Dhabi, UAE, and  University of Manchester, Manchester, M1 9BB, UK.   R. Schober is with the Institute for Digital Communications,
Friedrich-Alexander-University Erlangen-Nurnberg (FAU), Germany.  H. V. Poor is  with   Princeton University, Princeton, NJ 08544,
USA.   }
 
 \vspace{-2em}

  }\vspace{-5em}}
 \maketitle

\vspace{-1em}
\begin{abstract}
The aim of this paper is to develop hybrid non-orthogonal multiple access (NOMA)  assisted downlink transmission. First, for the    single-input single-output (SISO) scenario, i.e., each node is equipped with a single antenna, a novel hybrid NOMA  scheme is introduced, where NOMA is implemented as an add-on of a legacy time division multiple access (TDMA) network.  Because of the simplicity of    the SISO scenario, analytical results can be developed to reveal  important properties of downlink hybrid NOMA. For example,  in the case that the users' channel gains are ordered  and the durations of  their  time slots are the same, downlink hybrid NOMA is shown to always outperform TDMA, which is different from the existing conclusion for uplink hybrid NOMA. Second, the proposed downlink  SISO hybrid NOMA scheme is extended to  the     multiple-input single-output (MISO) scenario, i.e., the base station has multiple antennas.    For the MISO scenario,   near-field communication   is considered  to illustrate how NOMA can be used as an add-on in  legacy networks based on space division multiple access  and TDMA.  Simulation results   verify the developed analytical results and demonstrate the superior performance of downlink hybrid NOMA compared to conventional orthogonal multiple access.  
\end{abstract}\vspace{-0.5em}

\vspace{-1em}

 \section{Introduction}
 Multiple access techniques can be viewed as the foundation stone of modern mobile networks, since the design of many crucial  components of mobile networks, such as scheduling, resource allocation, channel estimation and signal detection, depends on which multiple access technique is used \cite{mojobabook}.  In the sixth-generation (6G) era, non-orthogonal multiple access (NOMA) has already received considerable  attention due to its superior spectral efficiency, compared to conventional orthogonal multiple access (OMA) \cite{you6g,10175525,10185552,9681865}.
 
Unlike   most existing works, e.g.  \cite{10218331,10194943,10041974,10158994,9028248},  which viewed 
 NOMA and OMA       as competing systems, this paper considers   NOMA as an add-on of OMA, which yields the following  two benefits \cite{9693536}. One is to shed light on the design of a unified framework for next-generation multiple access, and the other is to develop a vision for how    NOMA can be integrated into existing wireless systems,  which are based on OMA.     We note that   there are some existing works that  have recognized the importance of allowing   NOMA and OMA to co-exist. For example, in \cite{Zhiguo_CRconoma, 9693417,8823023,9352956,8641304}, user clustering has been carried out, where NOMA was implemented among the users within the same cluster and OMA was used to avoid inter-cluster interference. Similarly, in \cite{8895763,8618435,9525063},  various schemes  have been proposed to ensure that a user can intelligently  switch between the NOMA and OMA modes. While    these existing approaches  might realize a sophisticated coexistence between   NOMA and OMA, they cannot ensure that   NOMA is  used as a simple add-on of OMA, which causes major  disruptions to      OMA based  legacy networks, i.e.,  these NOMA approaches cannot be straightforwardly implemented in the currently deployed   OMA networks.
 
     \begin{figure}[t]\centering \vspace{-0em}
    \epsfig{file=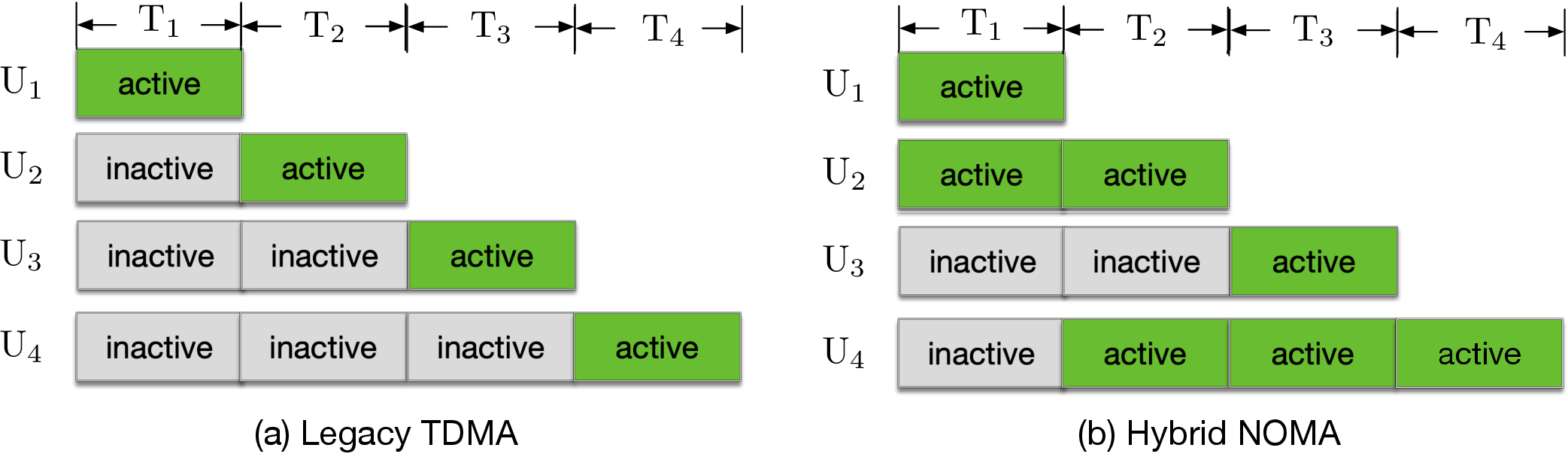, width=0.48\textwidth, clip=}\vspace{-0.5em}
\caption{Illustration of the key idea of hybrid NOMA transmission.    \vspace{-1em}    }\label{fig0}   \vspace{-1em} 
\end{figure}

Hybrid NOMA, a concept originally developed in mobile edge computing (MEC) networks \cite{9679390,9340353,Zhiguo_MEC1},  can   ensure that NOMA is implemented as an effective add-on of OMA.  The key idea of hybrid NOMA can be illustrated by the four-user example shown in Fig. \ref{fig0}. In particular, with conventional time division multiple access (TDMA), the four users, denoted by ${\rm U}_m$, $1\leq m \leq 4$, are served in four time slots individually,   as shown in Fig. 1(a). By using hybrid NOMA, the four users are still scheduled to finish their data transmission as in OMA, i.e., ${\rm U}_m$ finishes its transmission by the end of the $m$-th time slot, as shown in Fig. 1(b). Unlike TDMA, the use of   hybrid NOMA ensures that a user can also use   the time slots which are allocated to other users in TDMA, e.g.,   ${\rm U}_4$ can use not only its own TDMA time slot, i.e., the fourth time slot, but also the second and third time slots which belong exclusively to  ${\rm U}_2$ and ${\rm U}_3$ in TDMA. Because the users have more flexibility  to transmit, naturally   hybrid NOMA can outperform TDMA.    Time-slot allocation for hybrid NOMA can be effectively accomplished  by applying multi-time-slot power allocation. This is beneficial from the optimization perspective, since the optimal solution of power allocation is less challenging to obtain compared to time-slot allocation which is an   integer programming problem. We note that this multi-time-slot optimization feature makes hybrid NOMA different from the existing single-time-slot non-hybrid NOMA  schemes considered  in \cite{10218331,10194943,10041974,10158994,9693536}.


Because hybrid NOMA was originally developed for MEC, all   existing works on hybrid NOMA focused on uplink transmission. In particular, the use of hybrid NOMA   ensures that    multiple users can cooperate with each other for    offloading (or transmiting) their  computation tasks to a base station   \cite{9679390}. As shown in \cite{9679390,9340353,Zhiguo_MEC1}, if the energy consumption of   MEC offloading is used as the performance metric,  uplink hybrid NOMA yields better performance than pure NOMA, and for the two-user special case, hybrid NOMA is shown to outperform  OMA if one user's task deadline is less than two times of the other user's task deadline. Uplink hybrid NOMA has also been shown to outperform      OMA, if   energy efficiency is used as the performance metric \cite{9174768,7523951}. The performance of uplink hybrid NOMA can be further improved by applying intelligent reflecting surfaces and deep reinforcement learning as shown in \cite{9964376,9968198}. The use of uplink hybrid NOMA has also been shown to be beneficial  to improve the secrecy performance of MEC offloading \cite{9362298}.

Unlike the aforementioned uplink hybrid NOMA works,  this paper aims to design downlink hybrid NOMA transmission.  The contributions of this paper are listed as follows:
\begin{itemize}
\item For the single-input single-output (SISO) scenario, i.e., both the base station and the users are equipped with a single antenna, a new hybrid NOMA assisted downlink transmission scheme is proposed. In particular, it is assumed that there exists a TDMA based legacy network, and with hybrid NOMA, the users are encouraged to use the time slots which they would  not have access to in TDMA.   In addition, a multi-objective energy consumption minimization   problem is formulated, and solved by using successive resource allocation \cite{9679390}.

\item The properties of the obtained  power allocation solutions are analyzed to unveil   the important features of hybrid NOMA downlink transmission. For example, for the case that the users' channel gains are ordered   and the durations of  their time slots are the same, downlink hybrid NOMA is shown to always outperform OMA. This conclusion is different from the one previously reported for the uplink case   \cite{Zhiguo_MEC1}. In addition, the obtained analytical results show that it is optimal for each user to use the same transmit power across  all NOMA time slots, and the users' accumulated transmit powers on different time slots are  the  same. Furthermore, the solution obtained from successive resource allocation is shown to be a Pareto-optimal solution of the formulated multi-objective energy minimization problem. 

\item The developed downlink SISO hybrid NOMA scheme is then extended to   the multiple-input single-output (MISO) scenario, i.e., the base station has multiple antennas and each  user is equipped with a single antenna. Unlike the SISO network, the legacy downlink MISO network is based on both TDMA and space-division multiple access (SDMA).  In particular, it is assumed that there exist two groups of users, where the users in each group are served simultaneously via SDMA, and TDMA is used to avoid inter-group interference. In order to ensure the compatibility to the legacy network,  downlink   hybrid NOMA is used to realizes beam sharing, i.e.,    one group of users can use the spatial beams preconfigured for the other group of users.   In order to demonstrate the feasibility of hybrid NOMA,   near-field downlink transmission is considered  as an illustrative  example, as in general   the accurate    beamfocusing    in near-field communications can make beam sharing difficult. 

\item For the MISO scenario, an energy consumption minimization problem is first formulated.   Compared to the case of SISO downlink hybrid NOMA, the considered energy minimization problem in the MISO case is more challenging. This is due to the fact that beam sharing leads to potential inter-beam interference, which makes the formulated energy minimization problem non-convex. By applying successive convex approximation (SCA), a low-complexity sub-optimal power allocation solution is obtained and shown to outperform the OMA solution for various simulation setups. 
\end{itemize}

   The remainder  of this paper is organized as follows. In Section \ref{section siso}, the SISO scenario is considered, where the design of  downlink hybrid NOMA transmission is studied, and the properties of SISO hybrid  NOMA power allocation are unveiled.    In Sections \ref{section miso}, the MISO scenario is focused on, where the combination of SDMA and NOMA is investigated, and an energy consumption minimization problem is formulated and solved.  In Section \ref{section simulation},   simulation results are presented to demonstrate the performance of hybrid NOMA,  and the paper is concluded in Section \ref{section conclusion}.  
   
 \section{Downlink SISO Hybrid NOMA     Transmission}\label{section siso}
In this section, the application of hybrid NOMA in   SISO   downlink transmission  is studied.  Due to the simplicity of the SISO   scenario, an insightful understanding of the key features of  hybrid NOMA assisted     downlink transmission can be obtained  as will be shown in the following. 
 
 \subsection{Description of SISO Hybrid NOMA Transmission} 
 Consider  the scenario with a legacy  SISO TDMA  network with $M$ users, denoted by ${\rm U}_m$, where the base station serves ${\rm U}_m$ in the $m$-th time slot, denoted by ${\rm T}_m$. The key idea of hybrid NOMA is to encourage   spectrum sharing among the users, where a user can have access to multiple time slots, as shown in Fig. \ref{fig0}.  For illustrative  purposes, it is assumed that in ${\rm T}_i$, $(M-i+1)$ users, i.e., ${\rm U}_m$, $i\leq m \leq M$, are served simultaneously. For example, in ${\rm T}_1$, all $M$ users are served simultaneously, whereas in ${\rm T}_M$,  only ${\rm U}_M$ is served \footnote{How the users are scheduled depends on their quality of service requirements, the priority of the network    traffic,  and the features of the legacy network. For example, for MEC applications,   the users can be ordered according to the urgency of their computation tasks.  }. Therefore,  in ${\rm T}_i$, ${\rm U}_m$ receives the following signal: 
 \begin{align}
 y_{m,i} = h_{m}\sum^{M}_{j=i}\sqrt{P_{j,i}} s_{j,i} +n_{m,i},
 \end{align} 
 where $h_{m}$,  $P_{m,i}$,  $s_{m,i}$, and $n_{m,i}$ denote  ${\rm U}_m$'s channel gain, transmit power, transmit signal and received noise in ${\rm T}_i$, respectively.  Quasi-static fading is assumed, i.e., each user's channel gain remains constant within one time frame consisting of $M$ time slots. Furthermore,  the base station is assumed to have access to the users' channel state information (CSI). In practice, this CSI assumption can be realized by  asking each user to first perform channel estimation based on the pilot signals broadcasted by the base station, and then feed the CSI back to the base station via a reliable feedback channel.  
 
In ${\rm T}_i$, ${\rm U}_m$ observes the signals of $(M-i+1)$ users, and the following successive interference cancellation (SIC) is  carried out.    In particular, in each time slot, ${\rm U}_M$'s signal is always the first to be decoded, then    ${\rm U}_{M-1}$'s signal is decoded. In other words,  a descending decoding order is used, e.g., ${\rm U}_m$'s signal is decoded before ${\rm U}_j$'s, $m>j$. This descending decoding order ensures that  in ${\rm T}_m$,  ${\rm U}_m$'s signal can be decoded in the last SIC stage without interference, i.e., as if  ${\rm U}_m$   solely occupied  ${\rm T}_m$ via OMA.

  Given this SIC decoding order, in ${\rm T}_i$,  ${\rm U}_k$ can decode ${\rm U}_m$'s signal with the following achievable data rate \footnote{For notational simplicity, it is assumed that $\sum^{m-1}_{j=m}P_{j,i}=0$, and  the natural logarithm is used for the data rate expressions. }:
 \begin{align}
 R_{m,i}^k = \log\left(
 1+\frac{|h_{k}|^2P_{m,i}}{|h_{k}|^2\sum^{m-1}_{j=i}P_{j,i}+1 }
 \right),
 \end{align}
 for $i\leq k \leq m$ and $1\leq i \leq m$, 
 where the noise power is assumed to be normalized. As a result, by using hybrid NOMA, the achievable data rate of ${\rm U}_m$ at  ${\rm T}_i$ is give by
  \begin{align}\label{rmi}
 R_{m,i} = \min \left\{
  R_{m,i}^i, \cdots,  R_{m,i}^m
 \right\}.
 \end{align}
 {\it Remark 1:}  Downlink hybrid NOMA is a general framework, of which  conventional OMA is a   special case. In particular,  ${\rm U}_m$ can choose $P_{m,i}= 0$, $i\leq m$, and $P_{m,m}\neq 0$. This means that  ${\rm U}_m$ uses ${\rm T}_m$ only   and hence   pure OMA   is adopted   since ${\rm T}_m$ is allocated to ${\rm U}_m$ in OMA, as illustrated in Fig. 1(a).   By adjusting the power allocation coefficients, $P_{m,i}$, the use of  downlink hybrid NOMA can ensure that each user  fully  benefits    from the advantages of both NOMA and OMA transmissions. 

  \subsection{Power Allocation for SISO Hybrid NOMA Transmission }
 The users' energy consumption will be used as the metric for   performance evaluation, as explained in the following. With  downlink  hybrid NOMA, each user can have access to multiple time slots, but     an improper use of these time slots  can lead to a  surge in energy consumption, e.g., when a user's transmit power was mistakenly chosen to be large in a time slot when its channel condition is poor.

 The   energy minimization problem considered in this paper can be formulated as follows:
  \begin{problem}\label{pb:1} 
  \begin{alignat}{2}
\underset{P_{m,i}\geq0  }{\rm{min}} &\quad     \mathbf{E} \triangleq  \begin{bmatrix}
E_1 &\cdots &E_M
\end{bmatrix}^T  \label{1tst:1}
\\ s.t. &\quad   \sum^{m}_{i=1}R_{m,i}T\geq N_b,    1\leq m \leq M, \label{1tst:2} 
  \end{alignat}
\end{problem}  
where it is assumed that each user needs to receive  the same amount of data nats, denoted by $N_b$,     $E_m = \sum^{m}_{i=1}P_{m,i}T$, and $T$ denotes the duration of each time slot. 

Problem \eqref{pb:1} is a multi-objective optimization problem, and hence challenging to  solve. As pointed out in \cite{9679390}, the successive nature of SIC can be used to carry out successive resource allocation, as described  in the following. In particular, ${\rm U}_1$'s power allocation coefficient, $P_{1,1}$, is first optimized, by assuming that the other users' parameters are fixed. Then,   ${\rm U}_m$'s power allocation coefficients, $P_{m,i}$, $1\leq i \leq m$, can be optimized by assuming ${\rm U}_j$'s coefficients are fixed, $j>i$. The optimality of successive resource allocation will be discussed later after the closed-form expressions of the optimal power allocation coefficients are obtained. 

By using this successive resource allocation approach, ${\rm U}_m$'s power allocation coefficients can be obtained by solving the following simplified optimization problem:
 \begin{problem}\label{pb:1.5} 
  \begin{alignat}{2}
\underset{P_{m,i}\geq0   }{\rm{min}} &\quad     \sum^{m}_{i=1}P_{m,i} 
  \label{1.5tst:1}
\\ s.t. &\quad   \sum^{m}_{i=1}R_{m,i}\geq R , \label{1.5tst:2} 
  \end{alignat}
\end{problem}  
where $R=\frac{N_b}{T}$. It is straightforward to show that problem \eqref{pb:1.5} is a convex optimization problem, and hence   can be solved by using off-shelf optimization solvers. However, it is challenging to obtain a closed-form expression for the optimal solution of problem \eqref{pb:1.5}, mainly due to the dynamic nature of hybrid NOMA power allocation. For example, ${\rm U}_m$ might choose to transmit in a few non-consecutive time slots, and keep silent in the other time slots. 

\subsection{Properties  of SISO Hybrid NOMA Power Allocation}
In order to obtain an insightful understanding  of the properties  of  downlink hybrid NOMA,  two     power allocation solutions  for two special cases are provided. The first case  is 
 the pure OMA solution,   $ P_{m,m}^{\rm O}   =\frac{e^R-1}{|h_{m}|^2} $ and $ P_{m,i}^{\rm O}=0$ for $i<m$.  We note that the OMA solution is always a feasible solution of problem \eqref{pb:1.5}, but not necessarily the optimal solution. 
The   solution for the second case of interest is presented in the following lemma.  

\begin{lemma}\label{lemma1}
Without the constraint  of $ P_{m,i}^{\rm H}    \geq 0$,  an   optimal solution for the energy minimization problem shown in \eqref{pb:1.5} is given by
 \begin{align}\label{phybrid}
 P_{m,i}^{\rm H}  
 =&\left(\frac{e^R}{   \prod^{m}_{p=1}\frac{|\bar{h}_{m}|^2}{|\bar{h}_{m}|^2\sum^{m-1}_{j=p}P_{j,p} +1}}\right)^{\frac{1}{m}} \hspace{-1em} - \sum^{m-1}_{j=i}P_{j,i}^{\rm H}  -\frac{ 1}{  |\bar{h}_{m,i}|^2},
\end{align} 
where $|\bar{h}_{m,i}|^2=  \min \left\{
  | {h}_{i}|^2, \cdots,  | {h}_{m}|^2
 \right\}$.
\end{lemma}

\begin{proof}
See Appendix \ref{proof1}. 
\end{proof}

Because  Lemma \ref{lemma1} is obtained by omitting the constraint  $ P_{m,i}     \geq 0$, it is possible that $ P_{m,i}^{\rm H} \leq 0 $, i.e.,  Lemma \ref{lemma1} cannot be used for the general case. Instead,  off-shelf optimization solvers should be used for the general case to find the optimal solution of problem \eqref{pb:1.5}. The remainder of the section is to show that $ P_{m,i}^{\rm H}    > 0$ holds in \eqref{phybrid} for the special case, when  the users' channel gains are ordered.  The fact  that $ P_{m,i}^{\rm H}    > 0$ is significant since it means that Lemma \ref{lemma1} yields the optimal solution of problem \eqref{pb:1.5}, and downlink  hybrid NOMA outperforms OMA.     To facilitate the performance analysis, an important    feature of hybrid NOMA power allocation is established  first.

\begin{lemma}\label{lemma2}
Consider  the special case when the users are ordered according to their channel gains (i.e., $|h_m|^2>|h_{m+1}|^2$). If the  downlink   hybrid NOMA power allocation  in \eqref{phybrid} is adopted by ${\rm U}_i$, $1\leq i\leq m$,  the following equality holds
\begin{align}\label{lemmax1}
 \sum^{m}_{i=1}P_{i,1}^{\rm H}=\cdots =  \sum^{m}_{i=m-1}P_{i,m-1}^{\rm H}=P_{m,m}^{\rm H}.
 \end{align}
\end{lemma}
\begin{proof}
See Appendix \ref{proof2}.
\end{proof}

{\it Remark 2:} The term $\sum^{m}_{i=j}P_{i,j}^{\rm H}$ can be viewed as the accumulated interference in the $j$-th time slot. Lemma \ref{lemma2} reveals  that  the use of downlink  hybrid NOMA ensures that the accumulated interference in different time slots is identical. A   conclusion 
similar to Lemma \ref{lemma2} has been previously reported  for hybrid NOMA uplink transmission \cite{9679390,9340353,Zhiguo_MEC1}. 

With the help of Lemma \ref{lemma2}, the optimality  of the  downlink  hybrid NOMA power allocation in Lemma \ref{lemma1} can be established as shown in the following lemma. 

\begin{lemma}\label{lemma3}
For  the considered special case with ordered channel gains,  the solution shown in Lemma \ref{lemma1} is   an optimal solution of problem \eqref{pb:1.5}. 
\end{lemma}
\begin{proof}
See Appendix \ref{proof3}.
\end{proof}
We note that Lemma \ref{lemma3} is not sufficient to show the superiority of  downlink hybrid NOMA over OMA, since the OMA solution could also be an optimal solution of problem \eqref{pb:1.5}. 
\begin{lemma}\label{lemma4}
For  the considered special case with ordered channel gains,  there exists a single optimal solution for problem \eqref{pb:1.5}. 
\end{lemma}
\begin{proof}
See Appendix \ref{proof4x}.
\end{proof}

Based on Lemmas  \ref{lemma3} and \ref{lemma4},   the following corollary can be obtained straightforwardly. 
 \begin{corollary}\label{corollary0}
For  the considered special case with ordered channel gains,   downlink   hybrid NOMA always outperforms   OMA. 
 \end{corollary}

{\it Remark 3:} Corollary \ref{corollary0} shows that there is a unique difference between uplink and downlink hybrid NOMA. In particular, for the scenario considered in Corollary  \ref{corollary0}, uplink OMA   outperforms   uplink hybrid NOMA, as illustrated by the following two-user example. Similar to problem \eqref{pb:1}, an energy minimization problem for hybrid NOMA uplink transmission can be formulated as follows:
  \begin{problem}\label{pb:1c} 
  \begin{alignat}{2}
\underset{ }{\rm{min}} &\quad      P_{2,1}+P_{2,2} \label{1tst:1c}
\\ s.t. &\quad    \log\left(1+\frac{|h_2|^2P_{2,1}}{|h_1|^2P_{1,1}+1}
\right) +\log\left(
1+|h_2|^2P_{2,2}
\right)\geq R\label{1tst:2c} .
  \end{alignat}
\end{problem} 
By applying the Karush–Kuhn–Tucker (KKT) conditions, the  optimal solutions of problem \eqref{pb:1c}  are give by \cite{Boyd}
\begin{align}
P_{2,1} =& \left(\frac{e^R\left(|h_1|^2P_{1,1}+1\right)}{|h_2|^2 }\right)^{\frac{1}{2}}-\frac{|h_1|^2P_{1,1}+1}{|h_2|^2},
\\\nonumber
P_{2,2} = &  \left(\frac{e^R\left(|h_1|^2P_{1,1}+1\right)}{|h_2|^2 }\right)^{\frac{1}{2}}-\frac{ 1}{|h_2|^2}. 
\end{align}
It is straightforward to show that $P_{2,1}=0$ since
\begin{align}
P_{2,1} =& \left(\frac{e^Re^R}{|h_2|^2}\right)^{\frac{1}{2}}-\frac{e^R}{|h_2|^2} =0,
\end{align}
i.e., ${\rm U}_2$ chooses the OMA mode, where the last step follows from the fact that $P_{1,1}$ needs to satisfy the following equality:  $\log\left(|h_1|^2P_{1,1}+1\right)=R$.   As pointed out in \cite{9679390,9340353,Zhiguo_MEC1}, uplink hybrid NOMA can achieve a significant performance gain over OMA, if the durations of different   time slots are different. This conclusion does not contradict   the one made in this paper, since  Corollary  \ref{corollary0} is obtained by assuming   that the durations of all time slots are the same.   

 Lemma \ref{lemma3} shows that   the solution provided  in Lemma \ref{lemma1} is the optimal solution of problem \eqref{pb:1.5}. By using this conclusion and     following   steps similar to those in the proof of     \cite[Lemma 6]{9679390}, the following corollary can be obtained. 
\begin{corollary}
For  the considered special case with ordered channel gains,  the downlink  hybrid NOMA power allocation solution obtained in Lemma \ref{lemma1} is a Pareto optimal solution of the multi-objective optimization problem shown in \eqref{pb:1}. 
\end{corollary}
\section{Downlink MISO Hybrid NOMA Transmission}\label{section miso}
In this section, the  downlink SISO hybrid NOMA scheme developed in the previous section is extended to the MISO scenario. Unlike the SISO network, the legacy MISO network is based on both TDMA and SDMA.  In particular, it is assumed that there exist two groups of users, denoted by $\mathcal{G}_1$ and $\mathcal{G}_2$, respectively, where     the users in each group   are served simultaneously via SDMA, and TDMA is used to avoid inter-group interference, i.e., the users in $\mathcal{G}_1$ are scheduled to be served  earlier than those in $\mathcal{G}_2$. Similar to the SISO case, a user in $\mathcal{G}_2$ can  have access to  the   time slots which belong to $\mathcal{G}_1$. Unlike in the SISO case, the base station also needs to decide whether to design  new beamforming vectors for the users in   $\mathcal{G}_2$ during the   time slots which belong to the users in $\mathcal{G}_1$ \cite{9693536}. In order to avoid any disruptions to the legacy network,     beam sharing is used, i.e.,    the spatial beams preconfigured for the users in $\mathcal{G}_1$ are used to serve the users in $\mathcal{G}_2$. 

The importance of near-field communications in   future wireless networks motivates the use of   the near-field channel model    as an illustrative  example. We note that for conventional far-field beamforming, the concept of beam-sharing is straightforward, since many users  can share the same beam-steering vector and hence can be served by a single far-field beam \cite{mojobabook}. However,   the accurate   beamfocusing    in  near-field communications can make beam sharing difficult, which is another motivation for using near-field beamforming for  the feasibility study of hybrid NOMA in  downlink  MISO systems.
 \vspace{-1em}
 
 \subsection{Near-Field Communication System Model}
Consider a near-field MISO downlink   network, where a  base station  is equipped with  an $N$-antenna uniform linear array (ULA). There are $M$ and $K$ single-antenna users in $\mathcal{G}_1$ and and $\mathcal{G}_2$, respectively, which are denoted by    ${\rm U}_{m}^{\rm G1}$ and ${\rm U}_k^{\rm G2}$, respectively.  The scenario studied in  \cite{Dingnfhybrid1} is a special case with $K=1$.     Similar to \cite{Dingnfhybrid1},     the ULA is assumed  to be placed at the center of a  $2$-dimensional  plane. By using Cartesian coordinates,   the locations of  ${\rm U}_m^{\rm G1}$, ${\rm U}_k^{\rm G2}$,  the center of the ULA, and   the $n$-th element of the ULA are denoted by  ${\boldsymbol \psi}_m^{\rm G1}$, ${\boldsymbol \psi}_k^{\rm G2}$,  ${\boldsymbol \psi}_0$, and ${\boldsymbol \psi}_n$, respectively. 

 It is assumed that ${\rm U}_k^{\rm G1}$'s distance to the base station is much smaller than the Rayleigh distance. Therefore, the resolution of the beamformers of the  users in $\mathcal{G}_1$   is almost perfect, i.e., the users' channel vectors are almost orthogonal to each other,  which makes  the implementation of beamfocusing possible for the  users in $\mathcal{G}_1$  \cite{10123941,dingreso}. On the other hand,   the ${\rm U}_k^{\rm G2}$ could be   far-field users, or    near-field users whose distances to the base station are larger than those of the  users in $\mathcal{G}_1$.

 \vspace{-1em}
\subsection{An MSIO OMA  Benchmark }
It is assumed that ${\rm U}_k^{\rm G1}$ and  ${\rm U}_k^{\rm G2}$ cannot be    simultaneously served by SDMA, which   motivates the considered OMA benchmarking scheme   based on the combination of SDMA and TDMA\footnote{This assumption can be justified if  ${\rm U}_m^{\rm G1}$ and  ${\rm U}_k^{\rm G2}$ are near-field and far-field users, respectively, where there will be strong co-channel interference if the two types of users are served simultaneously. For the case that  all users are near-field users,  the assumption can be still justified, if  the resolution between  the channels of ${\rm U}_m^{\rm G1}$ and  ${\rm U}_k^{\rm G2}$ is not perfect, i.e., their channels are not perfectly  orthogonal to each other \cite{10123941,yuanweinear, dingreso}.  }.   In particular, the OMA transmission consists of two phases.  During the first phase which consists of     $M$ time slots, the base station serves the $M$   users   in $\mathcal{G}_1$ simultaneously via SDMA, and     the observation at ${\rm U}_m^{\rm G1}$ is given by 
 \begin{align}
 y^{\rm G1}_m = \mathbf{h}_m^H \sum^{M}_{m=1}\sqrt{P^{\rm G1}}\mathbf{w}^{\rm G1}_ms_m^{\rm G1}+n_m^{\rm G1},
 \end{align} 
where $n_m^{\rm G1}$ denotes the additive white Gaussian noise with normalized power,  $s_m^{\rm G1}$ denotes the symbol intended for ${\rm U}_m^{\rm G1}$, each      users in $\mathcal{G}_1$ is assumed to have the same transmit power, denoted by $P^{\rm G1}$, $\mathbf{w}^{\rm G1}_m$ denotes the user's beamforming vector,   the spherical channel model is adopted, i.e.,   $
 \mathbf{h}_m = \sqrt{N}\alpha_m^{\rm G1} \mathbf{b}\left({\boldsymbol \psi}^{\rm G1}_m\right)$, 
  $\mathbf{b}\left({\boldsymbol \psi} \right)=\frac{1}{\sqrt{N}}\begin{bmatrix} 
 e^{-j\frac{2\pi }{\lambda_w}\left| {\boldsymbol \psi}  -{\boldsymbol \psi}_1\right|} &\cdots &  e^{-j\frac{2\pi }{\lambda_w}\left| {\boldsymbol \psi}  -{\boldsymbol \psi}_N\right|}
 \end{bmatrix}^T$, $\lambda_w$ denotes the wavelength, and $\alpha_m^{\rm G1} = \frac{\lambda_w}{4\pi   \left| {\boldsymbol \psi}^{\rm G1}_m -{\boldsymbol \psi}_0\right|}$ \cite{yuanweinear,devar,9738442,Eldar2}.   Similar to the previous section, the base station is assumed to have access to the users' CSI. 
 Therefore,   ${\rm U}_m^{\rm G1}$'s data rate in OMA can be expressed as: 
 \begin{align}
  \bar{R}_m^{\rm G1 } = \log\left(1+\frac{P^{\rm G1}  |\mathbf{h}_m^H\mathbf{w}^{\rm G1}_m|^2}{P^{\rm G1} \sum^{M}_{i=1,i\neq m} |\mathbf{h}_m^H\mathbf{w}^{\rm G1}_i|^2+1}\right).
 \end{align}

The second phase of OMA transmission consists of    $K$ time slots, where the  users in $\mathcal{G}_2$, ${\rm U}_k^{\rm G2}$, are served simultaneously also via SDMA, and their data rates, denoted by $\bar{R}_k^{\rm G2 } $, are  similar to those of the      users in $\mathcal{G}_1$, i.e., 
\begin{align} \label{rk2}
\bar{R}_k^{\rm G2 }  =\log\left(1+\frac{P^{\rm G2}_{k,2} |\mathbf{g}_k^H\mathbf{w}^{\rm G2}_k|^2}{  \sum^{K}_{i=1,i\neq k} P^{\rm G2}_{i,2} |\mathbf{g}_k^H\mathbf{w}^{\rm G2}_i|^2+1}\right),
\end{align}
 where $\mathbf{g}_k$ and $\mathbf{w}_k^{\rm G2}$ are defined similar to their   counterparts for the  users in $\mathcal{G}_1$, and $P^{\rm G2}_{k,2}$ denotes the transmit power of ${\rm U}_k^{\rm G2}$ during the second phase.  Because  the use of  downlink hybrid NOMA does not affect the  users in $\mathcal{G}_1$,   we will focus on  optimizing the parameters for the    users in $\mathcal{G}_2$, i.e., $P^{\rm G1}$ is fixed and $P^{\rm G2}_{k,2}$ is to be optimized.  
 
 The assumption that  the  users in $\mathcal{G}_1$ are   very close to the base station makes beamfocusing possible to these users. In particular, the direction vector of a   user's spherical channel vector can be used as the user's beamforming vector, i.e., $\mathbf{w}_m^{\rm G1}=\mathbf{b}\left({\boldsymbol \psi}^{\rm G1}_m\right)$, and $|\mathbf{h}_i^H\mathbf{w}_m^{\rm G1}|^2\approx 0$, for $i\neq m$, since the resolution for the near-field users very close to the base station is almost perfect, i.e., $|\mathbf{b}\left({\boldsymbol \psi}^{\rm G1}_i\right)^H\mathbf{b}\left({\boldsymbol \psi}^{\rm G1}_m\right)|^2\approx 0$ \cite{dingreso}. Alternatively,   zero-forcing   beamforming vector can also be used by the  users in $\mathcal{G}_1$, which can ensure that $|\mathbf{h}_i^H\mathbf{w}_m^{\rm G1}|^2$, $i\neq m$, is strictly zero, i.e., there is no inter-beam interference.  For the  users in $\mathcal{G}_2$, conventional zero-forcing beamforming is used, since beamfocusing cannot be applied  due to their large distances to the base station. 
 
%
%

  \vspace{-1em}  
 \subsection{Downlink MISO Hybrid NOMA Transmission }
By using downlink hybrid NOMA, the ${\rm U}_k^{\rm G2}$ can have access to the time slots in the first phase, i.e.,  the base station broadcasts the following signals during the first $M$ time slots:
  \begin{align}
 \mathbf{x}^{\rm NOMA}=\sum^{M}_{m=1}\mathbf{w}_m^{\rm G1}\left(\sqrt{P^{\rm G1}}s_m^{\rm G1}+\sum^{K}_{k=1}s_{m,k}\sqrt{\tilde{P}^{\rm G2}_{m,k} } s_{m,k}^{\rm G2}\right),
 \end{align}
 where $s_{m,k}^{\rm G2}$ denotes the signal sent to ${\rm U}^{\rm G2}_k$   on $\mathbf{w}_m^{\rm G1}$, $s_{m,k}$ is an indicator, i.e., $s_{m,k}=1$ if  ${\rm U}^{\rm G2}_k$ uses $\mathbf{w}^{\rm G1}_m$, otherwise $s_{m,k}=0$,  and $  \tilde{P}^{\rm G2}_{m,k}$ is a   power allocation coefficient to be optimized. 
 
For illustration purposes, it is assumed that $M>K$, and each user in $\mathcal{G}_2$  selects a single beam. Denote the index of the beam used by ${\rm U}^{\rm G2}_k$ by $i^k$, e.g., $i^k=2$ means that  
  on beam $\mathbf{w}_{2}^{\rm G1}$,   ${\rm U}^{\rm G2}_k$ is active. Recall that during the first phase,  ${\rm U}^{\rm G1}_{i^k}$  is scheduled to be served exclusively on $\mathbf{w}_{i^k}^{\rm G1}$ in TDMA. Therefore, similar to the previous section, ${\rm U}^{\rm G1}_{i^k}$   carries  out SIC by first decoding the signal for ${\rm U}^{\rm G2}_k$ with the following data rate: 
 \begin{align}\label{r1m}
R_{{i^k}\rightarrow k} ^{\rm G1}=  \log\left(1+\frac{ |\mathbf{h}_{i^k}^H \mathbf{w}_{i^k}^{\rm G1}|^2
  {P}^{\rm G2}_{k,1}
}{ {\rm I}_{i^k}^{\rm G1}+ 1}\right),
 \end{align}  
 where  $ {P}^{\rm G2}_{k,1}=\tilde{P}^{\rm G2}_{i^k,k}$, and
 \begin{align}
  {\rm I}_{i^k}^{\rm G1}= P^{\rm G1} \sum^{M}_{m=1 } |\mathbf{h}_{i^k}^H\mathbf{w}^{\rm G1}_m|^2+  
 \sum _{ j\neq k} | \mathbf{h}_{i^k}^H \mathbf{w}_{i^j}^{\rm G1}|^2{P}^{\rm G2} _{j,1}.
\end{align}
 We note that in the considered MISO context, this SIC decoding order can also be justified by the fact that 
  the effective channel gain of ${\rm U}^{\rm G1}_{i^k}$, $ |\mathbf{h}_{i^k}^H \mathbf{w}_{i^k}^{\rm G1}|^2$,  is strong  since  $\mathbf{w}_{i^k}^{\rm G1}$ is tailored to the channel vector of ${\rm U}^{\rm G1}_{i^k}$.

Assume that all the  users in $\mathcal{G}_2$  have the same target data rate, denoted by $R$.  We note that if   $R^{\rm G1}_{{i^k}\rightarrow k}\geq R$,   ${\rm U}^{\rm G1}_{i^k}$  can successfully remove its partner's signal and decode its own signal in the same manner  as   in OMA, i.e., the data rates for the  users in $\mathcal{G}_1$ in NOMA and OMA are the same. 

On the other hand, during the first phase, ${\rm U}^{\rm G2}_k$   decodes its signal directly with the following data rate:
\begin{align}
R_{k} ^{\rm G2}=   \log\left(1+\frac{ | \mathbf{g}_k^H \mathbf{w}_{i^k}^{\rm G1}  |^2  {P}^{\rm G2}_{k,1}    }{  {\rm I}_k^{\rm G2}+ 1}\right) ,
\end{align}
where
\begin{align}
{\rm I}_k^{\rm G2}=&{P^{\rm G1}}
\sum^{M}_{m=1} \left|\mathbf{g}_k^H\mathbf{w}_m^{\rm G1}  
\right|^2 +
 \sum _{ j\neq k} |\mathbf{g}_k^H \mathbf{w}_{i^j}^{\rm G1}|^2  {P}^{\rm G2}_{j,1} 
 .
\end{align}
Therefore, in downlink MISO hybrid NOMA, the achievable data rate for ${\rm U}^{\rm G2}_k$ is  given by
\begin{align}
R_{k,1}^{\rm G2}=&\min\left\{
R_{k}^{\rm G2} , R^{\rm G1}_{{i^k} \rightarrow k} 
\right\}. 
\end{align}

During the second phase which consists of $K$ time slots, SDMA can be employed again to support the $K$    users in $\mathcal{G}_2$, which  means that, during the second phase,  the data rate of 
${\rm U}^{\rm G2}_k$ in downlink hybrid NOMA, denoted by $R^{\rm G2}_{k,2} $, is the same as that for OMA, i.e.,  
$R_{k,2}^{\rm G2}=\bar{R}_k^{G2}$ 
  shown in \eqref{rk2}.

Recall that during the first phase, the users in $\mathcal{G}_2$ can use the beams preconfigured to the users in $\mathcal{G}_1$. Hence,   beam selection, i.e., how to choose beam $\mathbf{w}_{i^k}^{\rm G1}$ for ${\rm U}^{\rm G2}_k$, is crucial for the performance of hybrid NOMA. If beamfocusing is used for beamforming, a   user's beamforming vector is aligned to its channel vector. According to \cite{10123941,dingreso}, the resolution of near-field beamforming can be almost perfect in the angle domain, but not necessarily  in the distance domain. Given the fact that ${\rm U}^{\rm G1}_m$ is close to the base station and ${\rm U}^{\rm G2}_k$ is far away, ${\rm U}^{\rm G2}_k$ can select the legacy beam whose angle of departure is closest to its own one, i.e., $i^k = \underset{m}{\arg} \min |\theta_k^{\rm G2}-\theta_m^{\rm G1} |$, where $\left(\theta_k^{\rm G2}, r_k^{\rm G2}\right)$ denotes the polar coordinates of  ${\rm U}^{\rm G2}_k$, and $\left(\theta_m^{\rm G1}, r_m^{\rm G1}\right)$ denotes the polar coordinates of  ${\rm U}^{\rm G1}_m$.  If zero-forcing based beamforming is used, ${\rm U}^{\rm G2}_k$ can select the beam on which the user's effective channel gain is the largest, i.e., $i^k = \underset{m}{\arg} \max |\mathbf{g}_k^H\mathbf{w}_m^{\rm G1}|^2$. 

\vspace{-1em}
\subsection{Downlink MISO Hybrid NOMA Power Allocation} 
Similar to the previous section, an energy minimization problem can be formulated as follows:
  \begin{problem}\label{pb:3} 
  \begin{alignat}{2}
\underset{ {P}^{\rm G2}_{k,i}\geq 0  }{\rm{min}} &\quad     \sum^{K}_{k=1}\left(MT {P}^{\rm G2}_{k,1}  + KT{P}^{\rm G2}_{k,2}\right)\label{3tst:1}
\\ s.t. &\quad MT R^{\rm G2}_{k,1} +KT R^{\rm G2}_{k,2}  \geq TK R,    1\leq k \leq K \label{3tst:2} ,
  \end{alignat}
\end{problem}  
where $K $ is used on the right-hand side of \eqref{3tst:2}   to highlight that $K$ time slots are used during the second phase. Similarly to the SISO scenario, the perfect knowledge of the users' CSI is assumed to be available at the base station. 
In order to gain a better understanding to the properties of problem \eqref{pb:3}, the users' data rate expressions need to be simplified as follows. 

First, $R^{\rm G2}_{k}$ can be written  as the following explicit function of the power allocation coefficients:  
\begin{align}
R^{\rm G2}_{k} =   \log\left(1+\frac{g_{k,k}    {P}^{\rm G2}_{k,1}   }{  
 \sum _{ j\neq k}g_{k,j} {P}^{\rm G2}_{j,1} 
 +b_k}\right) ,
\end{align}
where $g_{k,j}= \left|\mathbf{g}_k^H \mathbf{w}_{i^j}^{\rm G1}\right|^2 $, and $b_k =P^{\rm G1} 
\sum^{M}_{m=1}\left|\mathbf{g}_k^H\mathbf{w}_m^{\rm G1}  
\right|^2 + 1$. Similarly, both $R^{\rm G1}_{{i^k}\rightarrow k} $ and $R_{k,2}^{\rm G2}$ can be expressed as the following explicit functions of ${P}^{\rm G2}_{k,1}$ and ${P}^{\rm G2}_{k,2}$:
 \begin{align}\label{r1m2}
R_{{i^k}\rightarrow k}^{\rm G1} =&  \log\left(1+\frac{h_{k,k}{P}^{\rm G2}_{k,1}}{   
 \sum _{ j\neq k} {h}_{k,j} {P}^{\rm G2}_{j,1} +d_k}\right) ,\\\nonumber
 R_{k,2}^{\rm G2} =  &  \log\left(1+ \frac{c_{k,k} {P}^{\rm G2}_{k,2} }{\sum^{K}_{i=1,i\neq k}c_{k,i}{P}^{\rm G2}_{i,2}+1} \right) ,
 \end{align}
 where $h_{k,j} = |\mathbf{h}_{i^k}^H \mathbf{w}_{i^j}^{\rm G1}|^2$,      $d_k = P^{\rm G1}  \sum^{M}_{m=1 } |\mathbf{h}_{i^k}^H \mathbf{w}_{m}^{\rm G1}|^2 +1$,  and $c_{k,i}= \left|\mathbf{g}_{k}^H \mathbf{w}^{\rm G2}_i  \right|^2 $. 

\subsubsection{OMA Power Allocation}
By assuming that ${\rm U}^{\rm G2}_k$  does not use the time slots in the first phase, i.e., ${P}^{\rm G2}_{k,1}=0$, the considered hybrid NOMA optimization problem is degraded to a simple OMA case, as follows: 
 \begin{problem}\label{pb:5} 
  \begin{alignat}{2}
\underset{ {P}^{\rm G2}_{k,2}\geq 0  }{\rm{min}} &\quad     \sum^{K}_{k=1}  {P}^{\rm G2}_{k,2} \label{5tst:1}
\\ s.t. &\quad  K R^{\rm G2}_{k,2}  \geq  KR,    1\leq k \leq K \label{5tst:2} ,
  \end{alignat}
\end{problem}  
which can be recast by using the simplified expressions of $ R_{k,2}^{\rm G2}$ as follows: 
 \begin{problem}\label{pb:6} 
  \begin{alignat}{2}
\underset{ {P}^{\rm G2}_{k,2}\geq 0  }{\rm{min}} &\quad     \sum^{K}_{k=1}  {P}^{\rm G2}_{k,2} \label{6tst:1}
\\ s.t. &\quad    \frac{c_{k,k} {P}^{\rm G2}_{k,2} }{\sum^{K}_{i=1,i\neq k}c_{k,i}{P}^{\rm G2}_{i,2}+1}   \geq e^{R}-1,    1\leq k \leq K \label{6tst:2} .
  \end{alignat}
\end{problem}  
It is straightforward to show that problem \eqref{pb:6} is a convex optimization problem, and hence can be solved efficiently by applying   off-shelf optimization solvers. 

\subsubsection{ Downlink Hybrid NOMA Power Allocation}
While the power allocation problem can be easily solved for the OMA case, it is more challenging to solve the general downlink hybrid NOMA problem which can be rewritten as follows: 
  \begin{problem}\label{pb:7} 
  \begin{alignat}{2}
\underset{ {P}^{\rm G2}_{k,i} \geq 0 }{\rm{min}} &\quad    \sum^{K}_{k=1}\left(M {P}^{\rm G2}_{k,1}  + K{P}^{\rm G2}_{k,2}\right)
\\ s.t. &\quad \frac{M}{K} \log\left(1+\frac{h_{k,k}{P}^{\rm G2}_{k,1}}{   
 \sum _{ j\neq k} {h}_{k,j} {P}^{\rm G2}_{j,1} +d_k}\right) \label{7tst:2} + \\&  \quad    \log\left(1+ \frac{c_{k,k} {P}^{\rm G2}_{k,2} }{\sum^{K}_{i=1,i\neq k}c_{k,i}{P}^{\rm G2}_{i,2}+1} \right) \geq  R,  1\leq k \leq K \nonumber 
\\
&\quad \frac{M}{K}  \log\left(1+\frac{g_{k,k}   {P}^{\rm G2}_{k,1}    }{  
 \sum _{ j\neq k}g_{k,j}  {P}^{\rm G2}_{j,1} 
 +b_k}\right) + \label{7tst:3}  \\&  \quad  \log\left(1+ \frac{c_{k,k} {P}^{\rm G2}_{k,2} }{\sum^{K}_{i=1,i\neq k}c_{k,i}{P}^{\rm G2}_{i,2}+1} \right)  \geq  R,    1\leq k \leq K .\nonumber
  \end{alignat}
\end{problem}
Unlike the OMA problem in \eqref{pb:6}, problem \eqref{pb:7} is non-convex, mainly due to the fact that the optimization variables appear in both the numerators and the denominators of the fractions in \eqref{7tst:2} and \eqref{7tst:3}.

 In the literature, SCA has been shown to be  effective for tackling  the non-convex constraints shown  in  \eqref{7tst:2} and \eqref{7tst:3} \cite{6675875}. To facilitate the implementation of SCA, the following vectors are defined:   $\mathbf{p}=\begin{bmatrix}{P}^{\rm G2}_{1,1} &\cdots & {P}^{\rm G2}_{K,1}
\end{bmatrix}^T$ and $\mathbf{e}=\begin{bmatrix}{P}^{\rm G2}_{1,2} &\cdots & {P}^{\rm G2}_{K,2}
\end{bmatrix}^T$, which means that constraint  \eqref{7tst:3}    can be   rewritten  as follows: 
  \begin{align}
&M  \log\left( 
\bar{\mathbf{g}}_k^T\mathbf{p}
 +b_k\right)-M  \log\left(  
\tilde{\mathbf{g}}_k^T\mathbf{p}
 +b_k \right)   \nonumber \\&  +K\log\left(\mathbf{c}_k^T\mathbf{e}+1 \right)-K\log\left( \tilde{\mathbf{c}}_k^T\mathbf{e}+1 \right)\geq  KR,  
\end{align}
where $\bar{\mathbf{g}}_k=\begin{bmatrix}g_{k,1}  &\cdots &g_{k,K} 
\end{bmatrix}^T$,   the $k$-th element of  $\tilde{\mathbf{g}}_k$ is zero,     the other elements of $\tilde{\mathbf{g}}_k$ are the same as those of $\bar{\mathbf{g}}_k$, $\mathbf{c}_k=\begin{bmatrix}c_{k,1}  &\cdots &c_{k,K} 
\end{bmatrix}^T$, and $\tilde{\mathbf{c}}_k$ is constructed in a similar manner as  $\tilde{\mathbf{g}}_k$ by using the elements of $ {\mathbf{c}}_k$. 

Therefore, at the $i$-th iteration of SCA, constraint  \eqref{7tst:3} can be approximated as follows: 
 {\small  \begin{align}\label{consdd2}
&\frac{M}{K}  \log\left( 
\bar{\mathbf{g}}_k^T\mathbf{p}
 +b_k\right)-\frac{M}{K} \log\left(  
\tilde{\mathbf{g}}_k^T\mathbf{p}_{i-1}
 +b_k \right)  +\frac{M}{K}  \frac{\tilde{\mathbf{g}}_k^T(\mathbf{p}-\mathbf{p}_{i-1})}{\tilde{\mathbf{g}}_k^T\mathbf{p}
 +b_k}  \nonumber \\& + \log\left(\mathbf{c}_k^T\mathbf{e}+1 \right)-\log\left( \tilde{\mathbf{c}}_k^T\mathbf{e}_{i-1}+1 \right)+ \frac{\tilde{\mathbf{c}}_k^T(\mathbf{e}-\mathbf{e}_{i-1})  }{ \tilde{\mathbf{c}}_k^T\mathbf{e}_{i-1}+1}\geq  R, 
\end{align}}
where $\mathbf{p}_{i-1}$ and $\mathbf{e}_{i-1}$ are obtained from the $(i-1)$-th iteration.  
Similarly, constraint    \eqref{7tst:2} can be first recast as follows:
  \begin{align}
&M \log\left( \bar{\mathbf{h}}_k^T\mathbf{p}  +d_k \right)-M \log\left(    
 \tilde{\mathbf{h}}_k^T\mathbf{p}   +d_k \right) \nonumber \\& +K\log\left(\mathbf{c}_k^T\mathbf{e}+1 \right)-K\log\left( \tilde{\mathbf{c}}_k^T\mathbf{e}+1 \right)\geq  KR, 
\end{align}
where $ \bar{\mathbf{h}}_k=\begin{bmatrix}
 {h}_{k,1} &\cdots & {h}_{k,K}
\end{bmatrix}^T$, and $ \tilde{\mathbf{h}}_k$ is the same as  $ \bar{\mathbf{h}}_k$ except its $k$-th element being zero. Again, at the $i$-th iteration of SCA, constraint  \eqref{7tst:2}  can be approximated as follows:
 {\small \begin{align}
&\frac{M}{K} \log\left( \bar{\mathbf{h}}_k^T\mathbf{p}  +d_k \right)-\frac{M}{K} \log\left(    
 \tilde{\mathbf{h}}_k^T\mathbf{p}_{i-1}   +d_k \right)+\frac{M}{K} \frac{ \tilde{\mathbf{h}}_k^T(\mathbf{p} -\mathbf{p} _{i-1})}{    
 \tilde{\mathbf{h}}_k^T\mathbf{p}   +d_k } \nonumber \\& + \log\left(\mathbf{c}_k^T\mathbf{e}+1 \right)- \log\left( \tilde{\mathbf{c}}_k^T\mathbf{e}_{i-1}+1 \right)+  \frac{\tilde{\mathbf{c}}_k^T(\mathbf{e}-\mathbf{e}_{i-1})  }{ \tilde{\mathbf{c}}_k^T\mathbf{e}_{i-1}+1 }\geq  R . \label{consdd1}
\end{align}}
SCA can be carried out iteratively to obtain a suboptimal solution of problem \eqref{pb:7}, where  the $i$-th stage of the SCA algorithm  
requires solving   the following 
convex optimization problem:   \begin{problem}\label{pb:8} 
  \begin{alignat}{2}
\underset{\mathbf{p},\mathbf{e} \geq 0 }{\rm{min}} &\quad  M\mathbf{1}_K^T\mathbf{p}  +K  \mathbf{1}_K^T\mathbf{e} \label{1tst:1}
\\ s.t. &\quad 
\frac{M}{K}\log\left( \bar{\mathbf{h}}_k^T\mathbf{p}  +d_k \right)-\frac{M}{K} \log\left(    
 \tilde{\mathbf{h}}_k^T\mathbf{p}_{i-1}   +d_k \right)  \\\nonumber & +\frac{M}{K}\frac{ \tilde{\mathbf{h}}_k^T(\mathbf{p} -\mathbf{p} _{i-1})}{    
 \tilde{\mathbf{h}}_k^T\mathbf{p}   +d_k }+ \log\left(\mathbf{c}_k^T\mathbf{e}+1 \right)\\\nonumber & - \log\left( \tilde{\mathbf{c}}_k^T\mathbf{e}_{i-1}+1 \right)+  \frac{\tilde{\mathbf{c}}_k^T(\mathbf{e}-\mathbf{e}_{i-1})  }{ \tilde{\mathbf{c}}_k^T\mathbf{e}_{i-1}+1 }\geq  R ,  1\leq k \leq K  ,
\\
&\quad \frac{M}{K} \log\left( 
\bar{\mathbf{g}}_k^T\mathbf{p}
 +b_k\right)-\frac{M}{K}  \log\left(  
\tilde{\mathbf{g}}_k^T\mathbf{p}_{i-1}
 +b_k \right)   \\& +\frac{M}{K}  \frac{\tilde{\mathbf{g}}_k^T(\mathbf{p}-\mathbf{p}_{i-1})}{\tilde{\mathbf{g}}_k^T\mathbf{p}
 +b_k}  \nonumber+ \log\left(\mathbf{c}_k^T\mathbf{e}+1 \right)\\\nonumber &- \log\left( \tilde{\mathbf{c}}_k^T\mathbf{e}_{i-1}+1 \right)+  \frac{\tilde{\mathbf{c}}_k^T(\mathbf{e}-\mathbf{e}_{i-1})  }{ \tilde{\mathbf{c}}_k^T\mathbf{e}_{i-1}+1 }\geq  R,    1\leq k \leq K  ,
  \end{alignat}
\end{problem}
where $\mathbf{1}_K$ is a $K\times 1$ all-one vector.  
 
 \section{Simulation Results}\label{section simulation}
The performance of downlink hybrid NOMA transmission for SISO and MISO systems is evaluated   in the following two subsections, respectively. 
 
     \begin{figure}[t] \vspace{-0em}
\begin{center}
\subfigure[Special case with ordered channel gains]{\label{fig1a}\includegraphics[width=0.35\textwidth]{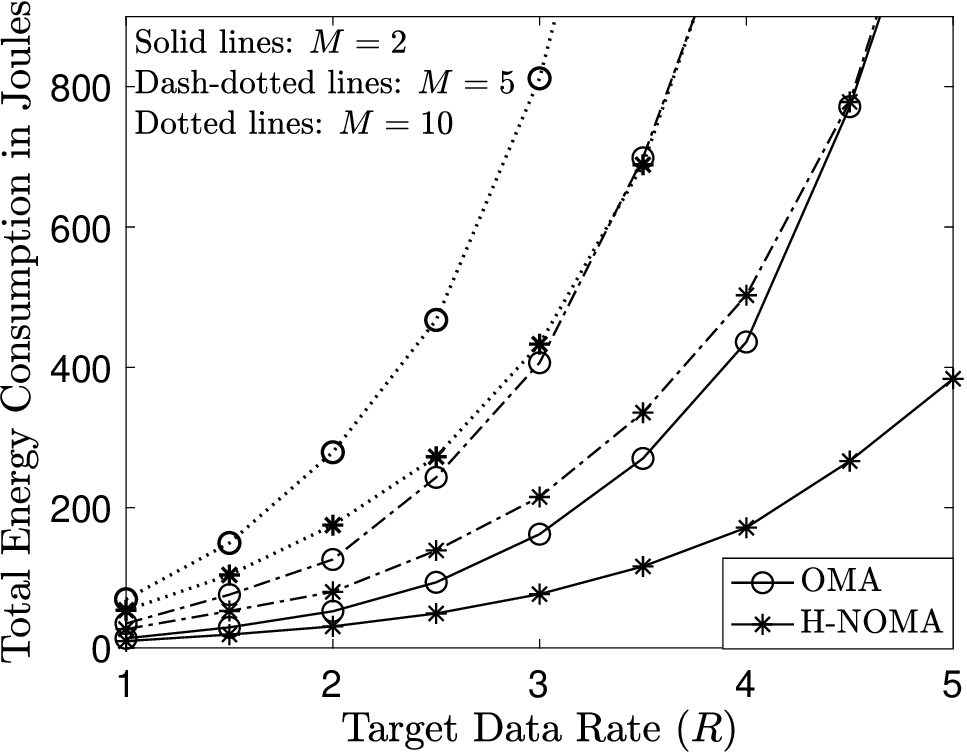}} 
\subfigure[ General case with unordered   channel gains]{\label{fig1b}\includegraphics[width=0.35\textwidth]{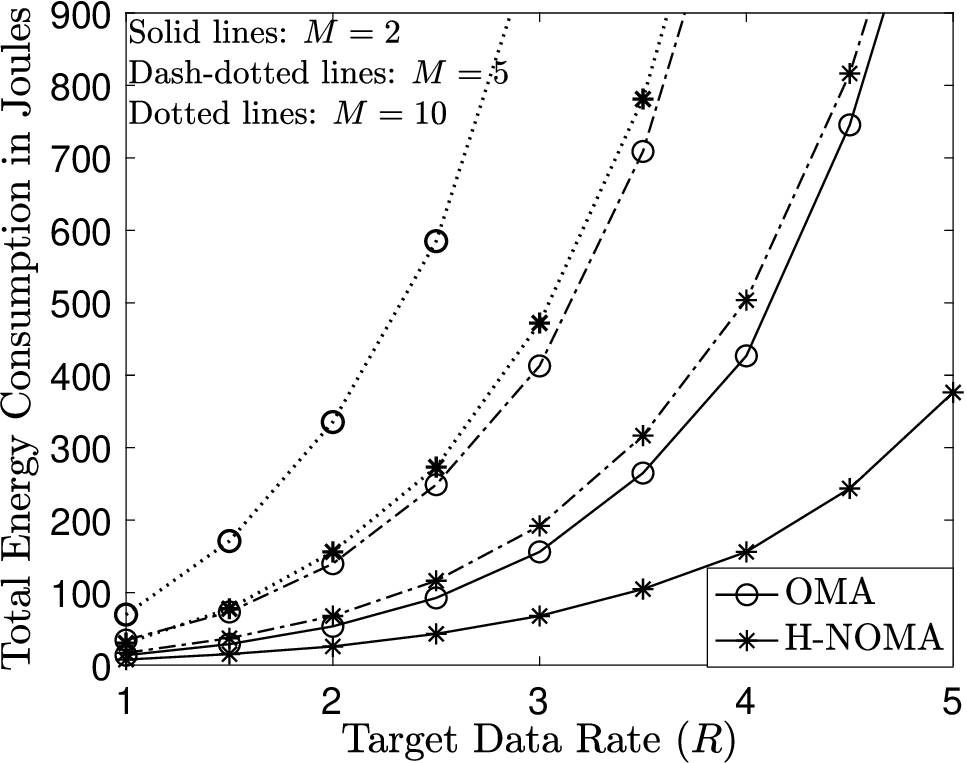}} \vspace{-1em}
\end{center}
\caption{ The total energy consumption realized by the considered transmission schemes in the SISO case.        \vspace{-1em} }\label{fig1}\vspace{-1em}
\end{figure}

     \begin{figure}[t] \vspace{-0em}
\begin{center} 
\subfigure[  Accumulated interference from the first $m$ users  in ${\rm T}_i$ - ${\rm I}_{m,i}=\sum^{m}_{j=i}P_{j,i}$]{\label{fig3a}\includegraphics[width=0.4\textwidth]{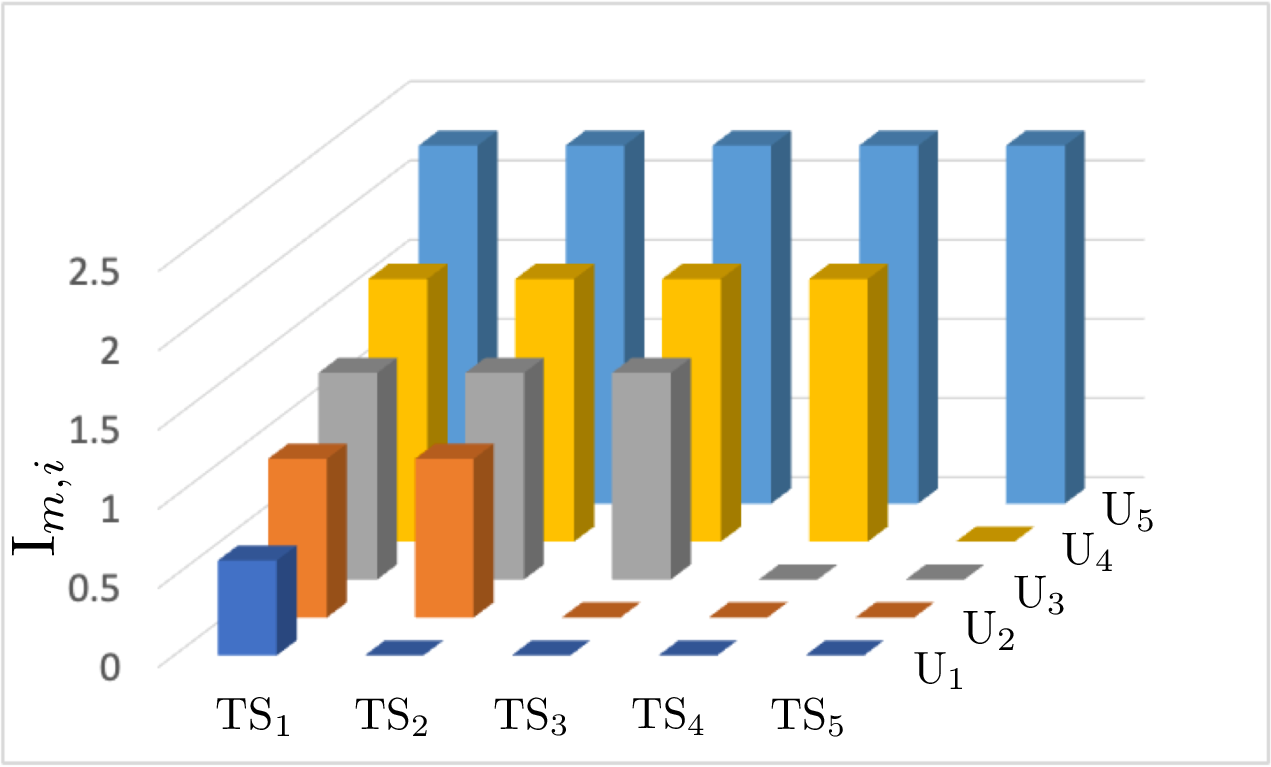}}
\subfigure[${\rm U}_m$'s power allocation in ${\rm T}_i$ - $P_{m,i}$]{\label{fig3b}\includegraphics[width=0.4\textwidth]{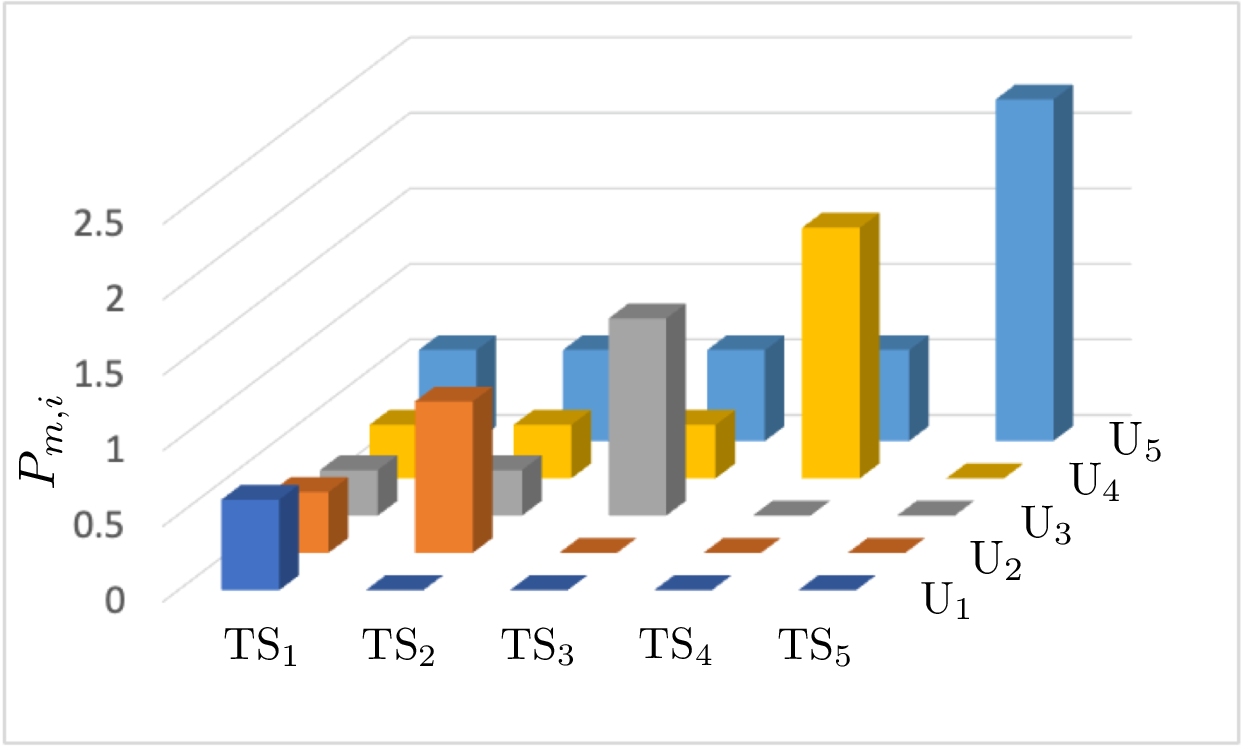}} \vspace{-1em}
\end{center}
\caption{  An illustration of the properties of hybrid NOMA power allocation. $h_m=m$, $1\leq m \leq M$,   and $R=2$ nats per channel use (NPCU).    \vspace{-1em} }\label{fig3}\vspace{-0em}
\end{figure}

 \subsection{Downlink SISO Hybrid NOMA Transmission}
 In Fig. \ref{fig1}, the total energy consumption realized by the proposed downlink  hybrid NOMA scheme is shown as a function of $R$, where the performance of OMA is also shown as a benchmark. For Fig. \ref{fig1}, the users' channels are assumed to be independent and identically distributed (i.i.d.) complex Gaussian random variables with zero mean and unit variance\footnote{With complex-Gaussian fading, the average   energy consumption can be infinite, as explained in the following. Take OMA as an example, where the average   energy consumption of    ${\rm U}_m^{\rm G1}$  is  $\mathcal{E}\left\{
\frac{e^{R}-1}{|h_{m}|^2}
\right\}$, which is   infinite since $\mathcal{E}\left\{
\frac{ 1}{x}
\right\}\rightarrow \infty$ for exponentially distributed $x$, where $\mathcal{E}\left\{ \cdot
\right\}$ denotes the expectation. Therefore,  in the simulation, a constraint of  $|h_{m}|^2\geq 0.01$ is imposed  to avoid this singularity issue.   }. Fig. \ref{fig1a} focused on  the special case that the users' channel gains are  ordered,  and   Fig. \ref{fig1b} focuses on the general case that the users' channel gains are not ordered. For the case with the ordered channel gains, the simulation results are obtained by applying optimization solvers to  solve problem \eqref{pb:1.5}, and the analytical results are obtained by applying the hybrid NOMA solution provided  in Lemma \ref{lemma1}. For the general case with unordered channel gains, only simulation results are presented.  The two  figures in Fig. \ref{fig1} show  that the use of  downlink hybrid NOMA can reduce the total energy consumption significantly, compared to OMA,  particularly for large $R$.  Fig. \ref{fig1a} also shows that the analytical results perfectly match  the simulation results, which verifies  Corollary \ref{corollary0}, i.e., for the special case with ordered channel gains, hybrid NOMA power allocation is the optimal solution of problem \eqref{pb:1.5} and  outperforms OMA. 

In Figs. \ref{fig3} and \ref{fig2}, the properties of hybrid NOMA power allocation are studied by focusing on a deterministic case with $h_m=m$, $1\leq m \leq M$. In particular, Fig. \ref{fig3a} shows that the accumulated interference from the first $m$ users  at different time slots is identical, i.e.,   $ \sum^{m}_{j=1}P_{j,1}=\cdots =\sum^{m}_{j=m-1}P_{j,m-1}= P_{m,m}$, which confirms Lemma \ref{lemma2}.
Fig. \ref{fig3b} shows the users' power allocation coefficients in different time slots, and demonstrates that ${\rm U}_m$ will use the same transmit power during the first $(m-1)$ time slots, i.e, the NOMA time slots. As discussed in the proof of Lemma \ref{lemma3}, the observation from Fig. \ref{fig3a} is the reason for the observation from Fig. \ref{fig3b}. Take ${\rm U}_5$ as an example.  Fig. \ref{fig3a} shows that the interferences experienced by ${\rm U}_5$ in the first $4$ time slots, i.e., the yellow bars, are   same. Furthermore, we note that  ${\rm U}_5$'s channel gains in these time slots are assumed to be the same, which means that  ${\rm U}_5$ naturally uses the same transmit powers during the first $4$ time slots. In Fig. \ref{fig2}, the hybrid NOMA power allocation solution   in Lemma \ref{lemma1} is shown to perfectly match the solution obtained by an exhaustive search, which verifies the closed-form expression of the optimal power allocation solution shown in  Lemma \ref{lemma1}.

    \begin{figure}[t]\centering \vspace{-0em}
    \epsfig{file=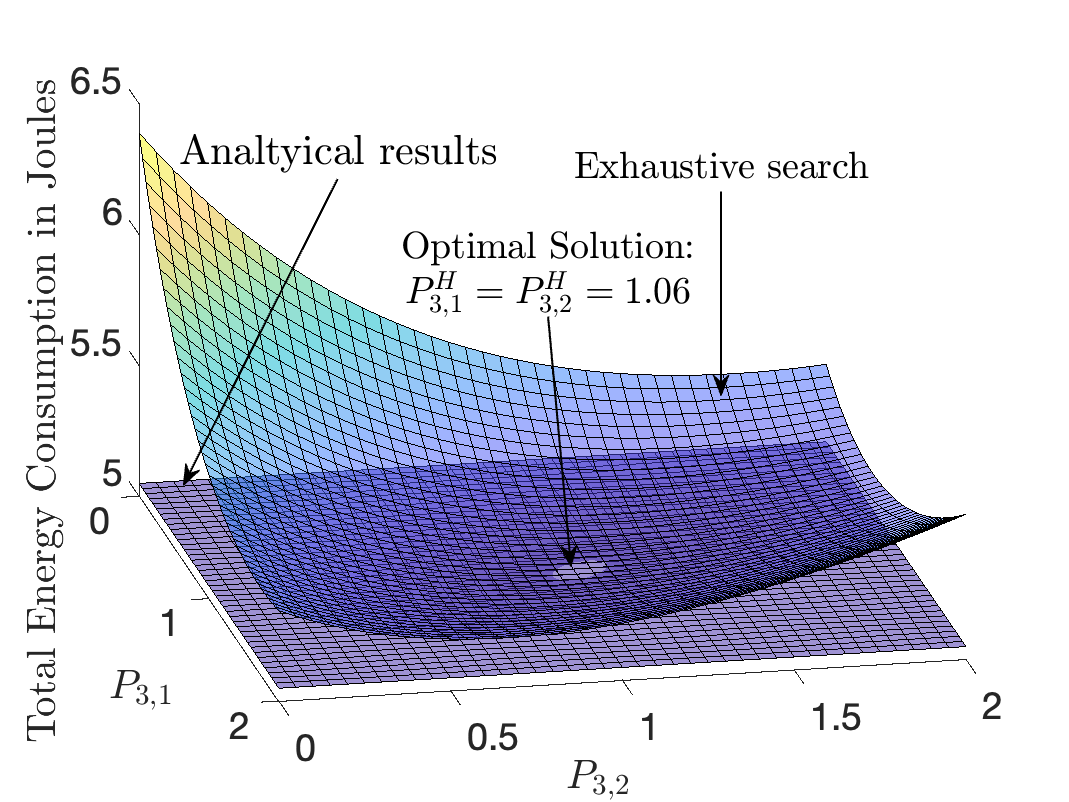, width=0.4\textwidth, clip=}\vspace{-0.5em}
\caption{Verification for the optimality of the hybrid NOMA solution shown in Lemma \ref{lemma1}, where   $h_m=m$, $1\leq m \leq M$, and $R=2$    NPCU. \vspace{-1em}    }\label{fig2}   \vspace{-1em} 
\end{figure}

 \subsection{Downlink MISO Hybrid NOMA Transmission}
 In this subsection, the performance of  downlink  MISO hybrid NOMA transmission is focused on, where the   carrier frequency is set to $28$ GHz, and the antenna spacing is half of the wavelength. In Fig. \ref{fig4}, the total energy consumption realized by the considered transmission schemes is illustrated, by assuming that the  users in $\mathcal{G}_1$ are  randomly located in the half-ring with radius $10$ m and $50$ m. The    users in $\mathcal{G}_2$, ${\rm U}_k^{\rm G2}$, are at  fixed locations     $200$ m away from the base station, and their angles are equally  spaced between $-\frac{\pi}{3}$ and $\frac{\pi}{3}$. The two  sub-figures of Fig. \ref{fig4} demonstrate that the use of  downlink hybrid NOMA can realize a significant performance gain over OMA, particularly if  the size of   $\mathcal{G}_1$ and the target data rate are large. Fig. \ref{fig4a} shows that the use of more antennas can reduce the energy consumption for both OMA and hybrid NOMA, where the performance gap between the two schemes is increased if there are   more    users in $\mathcal{G}_1$. Fig. \ref{fig4b} shows that when the target data rate is small, the performance gain of  downlink hybrid NOMA over OMA is small. This is due to the fact that for   small target data rates,   the time slots available in the second phase are sufficient to serve ${\rm U}_k^{\rm G2}$, and there is no need to employ NOMA and have access to the time slots in the first phase.  In addition, Fig. \ref{fig4} shows that   beamfocusing outperforms zero-forcing  based beamforming, which is consistent with the conclusion made in \cite{Dingnfhybrid1}.
 
   \begin{figure}[t] \vspace{-0em}
\begin{center}
\subfigure[$R=5$ NPCU]{\label{fig4a}\includegraphics[width=0.35\textwidth]{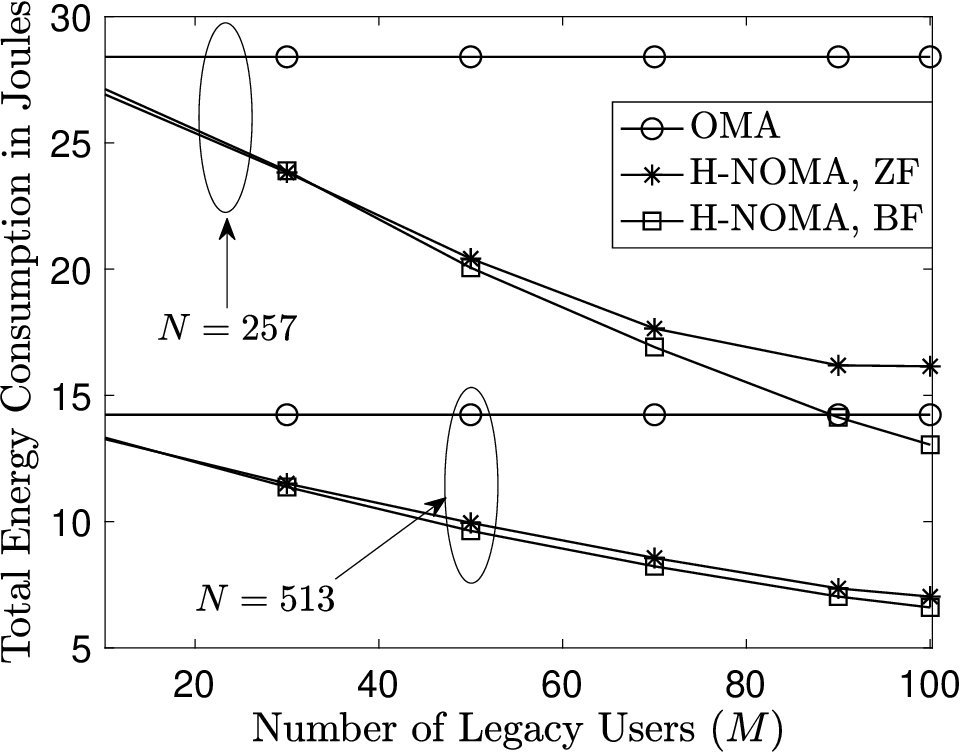}} 
\subfigure[$N=257$]{\label{fig4b}\includegraphics[width=0.35\textwidth]{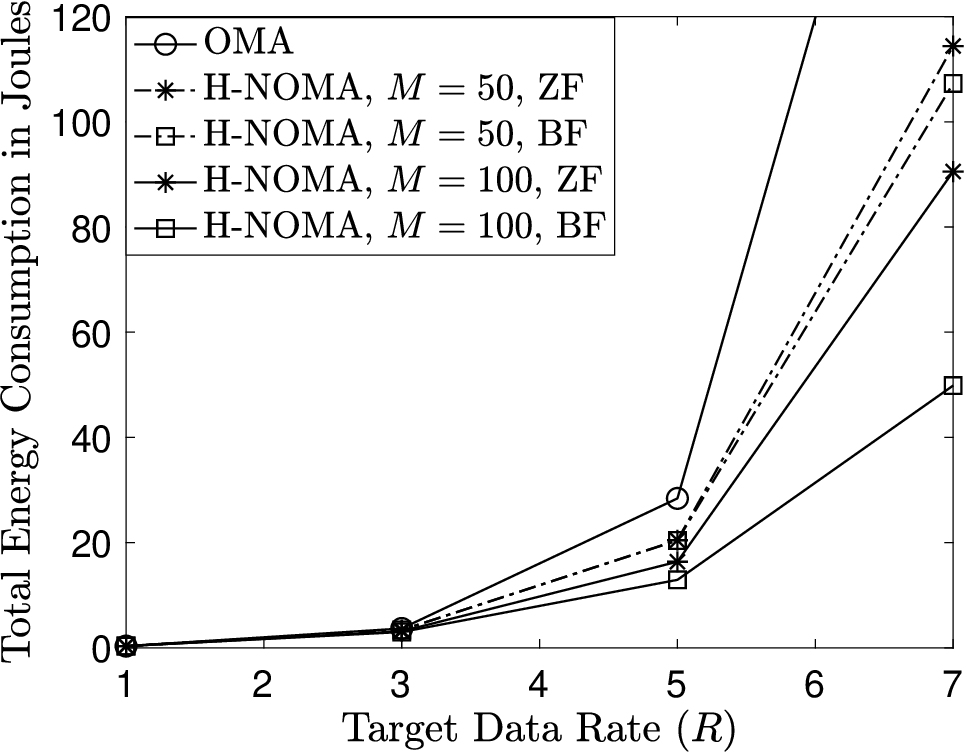}} \vspace{-1em}
\end{center}
\caption{The total energy consumption realized by the considered transmission schemes in the MISO case. The  users in $\mathcal{G}_1$, ${\rm U}_m^{\rm G1}$, are uniformly  located in the half-ring with radius $10$ m and $50$ m. The    users in $\mathcal{G}_2$, ${\rm U}_k^{\rm G2}$, are at fixed locations   $200$ m away from the base station, and their angles are equally  spaced between $-\frac{\pi}{3}$ and $\frac{\pi}{3}$.   $K=3$ and $P^{\rm G1}=10$ dBm.        \vspace{-1em} }\label{fig4}\vspace{-1em}
\end{figure}
 
    \begin{figure}[t] \vspace{-0em}
\begin{center} 
\subfigure[ $r^{\rm G1}=50$ m]{\label{fig5a}\includegraphics[width=0.35\textwidth]{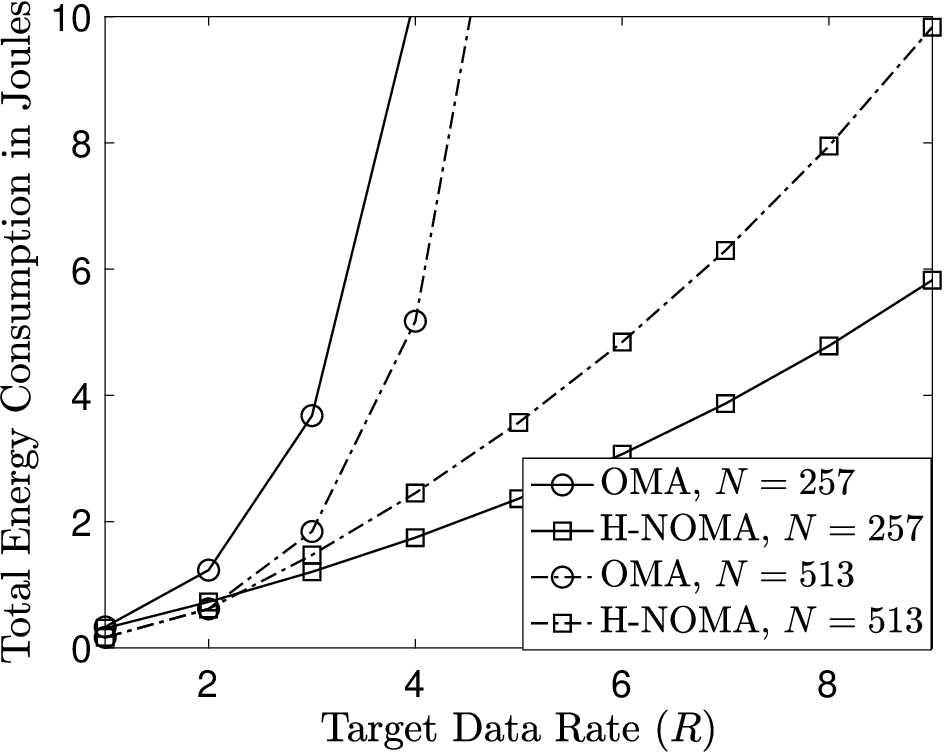}} \vspace{-1em}
\subfigure[ $r^{\rm G1}=100$ m]{\label{fig5b}\includegraphics[width=0.35\textwidth]{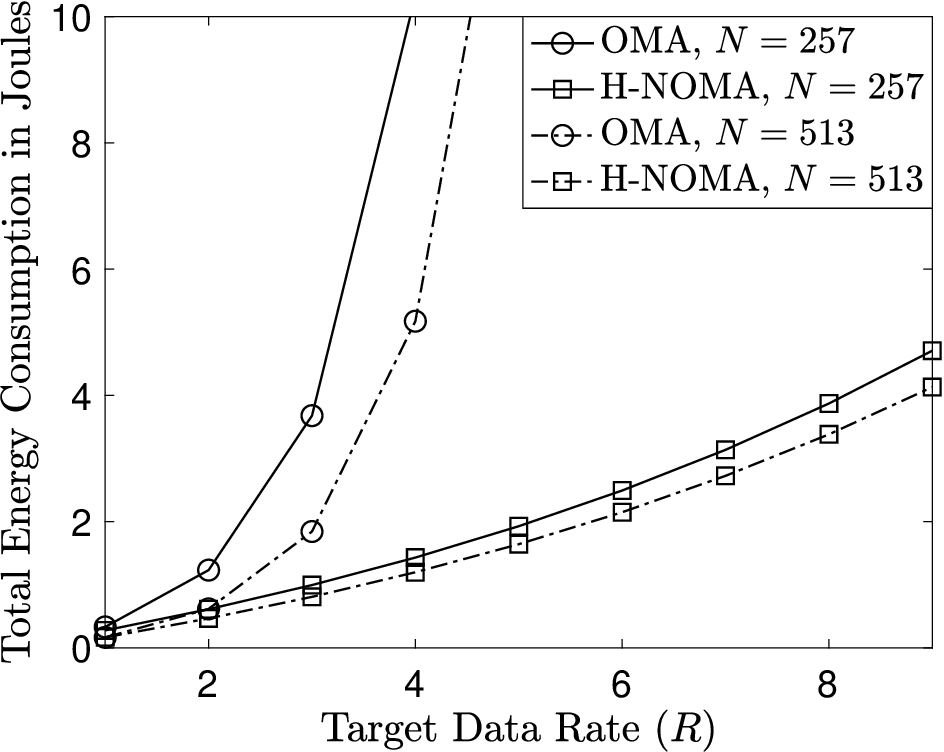}} 
\end{center}
\caption{ A deterministic study of  the total energy consumption realized by the considered transmission schemes. The    users in $\mathcal{G}_1$, ${\rm U}_m^{\rm G1}$,  are at   fixed locations  $r^{\rm G1}$ m away from the base station, and their angles are equally  spaced by $\frac{\pi}{2M}$.   The    users in $\mathcal{G}_2$, ${\rm U}_k^{\rm G2}$,  are $200$ m away from the base station, where $\theta_k^{\rm G1}=\theta_k^{\rm G2}$, $\left(\theta_k^{\rm G2}, r_k^{\rm G2}\right)$ denotes the polar coordinates of  ${\rm U}^{\rm G1}_m$, and $\left(\theta_m^{\rm G1}, r_k^{\rm G1}\right)$ denotes the polar coordinates of  ${\rm U}^{\rm G1}_M$.    $M=20$, $K=3$, and $P^{\rm G1}=10$ dBm.      \vspace{-1em} }\label{fig5}\vspace{-1em}
\end{figure}

\begin{table*}\vspace{-3em}
  \centering
  \caption{Total Energy Consumption for the Considered Deterministic Scenario of  $r^{\rm G1}=50$ m. }\vspace{-1em}
\label{table2x}
  \begin{tabular}{|c|c|c|c|c|c|c|c| }
\hline
R &  $ 1$ & $ 2$& $3$ & $4$ &$5$ &$6$&$7$ \\
    \hline
OMA, BF & $0.1692$  &$	0.6665$  &	$2.4081$   	& $ 22.5109$  	&$\infty $ &$\infty$&$\infty$  \\   
    \hline 
   H-NOMA, BF  & $0.1569$  	&$0.3684$  &$	0.6103$ 	&$0.8871$  	&$1.2043$  	&$1.5678$  	&$1.9849$  \\
    \hline
OMA, ZF & $0.1678$ 	&$0.6240$ 	&$1.864$ 	&$5.234$ 	&$14.397$ 	&$39.3052$ 	&$107.0106$   \\   
    \hline 
   H-NOMA, ZF  &$0.1568$ 	&$0.3675$ 	&$0.6075$  	&$0.8810$ 	&$1.1926$ 	&$1.5476$ 	&$1.9521$         \\\hline
  \end{tabular}\vspace{-1em}
\end{table*}

In Fig. \ref{fig5} and Table \ref{table2x}, a deterministic scenario is considered in order to reveal a few interesting properties of  downlink hybrid NOMA for   near-field communications.  In particular, the  users in $\mathcal{G}_1$, ${\rm U}_m^{\rm G1}$,  are at  fixed locations  $r^{\rm G1}$ m away from the base station, and their angles are equally  spaced by $\frac{\pi}{2M}$.  The    users in $\mathcal{G}_2$, ${\rm U}_k^{\rm G2}$,  are $200$ m away from the base station, where $\theta_k^{\rm G1}=\theta_k^{\rm G2}$.        Fig. \ref{fig5} is based on the use of beamfocusing, and  demonstrates that the use of hybrid NOMA can reduce the energy consumption of OMA significantly, which is consistent to the observations made in Fig. \ref{fig4}. Comparing the two   sub-figures of Fig. \ref{fig5}, an interesting observation is that for the case of $r^{\rm G1}=50$ m, the use of more antennas at the base station degrades the performance of  downlink hybrid NOMA. This is due to the fact that   the  users in $\mathcal{G}_1$ are very close to the base station, and hence for large numbers of antennas, the resolution of near-field beamforming   is almost perfect, i.e., the orthogonality among the  users' channel vectors is almost perfect, and one   user's beamfocusing based beam covers a small area centered at the   user. As a result, it is challenging for a  user in $\mathcal{G}_2$ to find a suitable legacy beam, which leads to the performance loss of  downlink hybrid NOMA. However, when $r^{\rm G1}$ is increased from $50$ m to $100$ m, the resolution of near-field beamforming becomes poor, which introduces the opportunity for   beam sharing among the users. 
In Table  \ref{table2x}, it is shown that for $r^{\rm G1}=50$ m,  beamfocusing and zero-forcing beamforming achieve practically  the same performance, which is again due to the fact that the  users in $\mathcal{G}_1$ are very close to the base station, and hence their channel vectors are almost orthogonal. An interesting observation from the table is that for OMA, the use of beamfocusing can lead to  infinite energy consumption, as explained in the following. Recall that in OMA,  ${\rm U}_k^{\rm G2}$ has to rely on the time slots in the second phase, and its achievable data rate is  given by $R_{k,2}^{\rm G2} =    \log\left(1+ \frac{c_{k,k} {P}^{\rm G2}_{k,2} }{\sum^{K}_{i=1,i\neq k}c_{k,i}{P}^{\rm G2}_{i,2}+1} \right)$. With beamfocusing, the inter-beam terms, $c_{k,i}$, are not zero, which means that there exists an upper bound on $R_{k,2}^{\rm G2}$. When the target data rate is larger than this upper bound, there is no  feasible solution for ${P}^{\rm G2}_{k,2}$ to realize the large target dada rate. However, by using hybrid NOMA, the user can have access to those time slots in the first phase as well, which avoids the singularity   suffered  by OMA and reveals an important advantage of  downlink hybrid NOMA. 

%
%


\section{Conclusions} \label{section conclusion}
In this paper,   hybrid NOMA assisted downlink transmission schemes have been developed for    SISO and MISO systems, respectively. For   downlink  SISO systems, analytical results were derived  to reveal several important properties of hybrid NOMA power allocation. For example,  in the case that users' channel gains are ordered  and the durations of all  time slots are the same,  downlink hybrid NOMA was shown to always outperform OMA, which is different from the conclusions obtained for uplink  hybrid NOMA   transmission.  For  downlink   MISO systems,   near-field communication   was considered  to illustrate how NOMA can be used as an add-on in   legacy networks based on SDMA and TDMA.  Simulation results were   presented to verify the developed analytical results and demonstrate the superior performance of  downlink  hybrid NOMA over conventional OMA. In this paper,   MISO hybrid NOMA was implemented between two groups of users, where an important direction for future research is to study how MISO  hybrid NOMA can be extended to the case with more than two groups of users. In addition, MISO hybrid NOMA was considered for  an ideal near-field system. Thus,  an important direction for future research is the consideration of the impact of practical issues, such as non-line-of-sight   paths and hybrid beamforming, on the design of  downlink hybrid NOMA. 
\appendices

\section{Proof for Lemma \ref{lemma1}}
\label{proof1}

It is straightforward to show that ${\rm U}_1$'s power allocation is   the same as that for OMA, i.e.,  $P_{1,1}=P_1^{\rm OMA}$. Therefore, the case of $m>1$ is focused on in this proof.  

Recall that  $R_{m,i} = \min \left\{
  R_{m,i}^i, \cdots,  R_{m,i}^m
 \right\}$.   We note that $f(x)\triangleq \frac{ax}{bx+1}$ is a monotonically increasing function of $x$, for positive $a$, $b$ and $x$. By using this observation   and the fact that   $|\bar{h}_{m,i}|^2=  \min \left\{
  | {h}_{i}|^2, \cdots,  | {h}_{m}|^2
 \right\}$,  the expression for $R_{m,i} $ can be simplified, and problem \eqref{pb:1.5} can be recast as follows: 
   \begin{problem}\label{pb:2} 
  \begin{alignat}{2}
\underset{P_{m,i}    }{\rm{min}} &\quad    \sum^{m}_{i=1}P_{m,i}  \label{2tst:1}
\\ s.t. &\quad   \sum^{m}_{i=1} \log\left(
 1+ b_{m,i}|\bar{h}_{m,i}|^2P_{m,i}\right) \geq R,    \label{2tst:2} 
  \end{alignat}
\end{problem}  
 where the constraints   $P_{m,i}\geq0$ are omitted, and  $b_{m,i}=\frac{1}{|\bar{h}_{m,i}|^2\sum^{m-1}_{j=i}P_{j,i} +1}$.
 
It is straightforward to show that problem \eqref{pb:2} is a convex optimization problem, and its optimal solution can be found by applying the KKT conditions. In particular, the Lagrangian  of problem \eqref{pb:2} is given by
\begin{align}\label{lagran}
L =&  \sum^{m}_{i=1}P_{m,i}  + \lambda \left(R- \sum^{m}_{i=1} \log\left(
 1+ b_{m,i}|\bar{h}_{m,i}|^2P_{m,i}\right)\right),
\end{align}
where $\lambda$ is the Lagrangian  multiplier.  
The derivative of the Lagrange  is given by
\begin{align}\label{kkt1x}
\frac{\partial L}{\partial P_{m,i}} = 1 - \lambda \frac{b_{m,i}|\bar{h}_{m,i}|^2}{    
 1+ b_{m,i}|\bar{h}_{m,i}|^2P_{m,i} }=0,
\end{align}
which means that $ P_{i,m}$ can be expressed as follows:
 \begin{align}\label{pim1}
 P_{m,i} =\lambda  -\frac{|\bar{h}_{m,i}|^2\sum^{m-1}_{j=i}P_{j,i} +1}{  |\bar{h}_{m,i}|^2}.
\end{align}
The expression for  $\lambda$ can be found by using the following equality: 
$  \sum^{m}_{i=1} \log\left(
 1+ b_{m,i}|\bar{h}_{m,i}|^2P_{m,i}\right)=R$,
which yields
\begin{align}\label{lambda}
 \lambda  =\left(\frac{e^R}{   \prod^{m}_{i=1} b_{m,i}|\bar{h}_{m,i}|^2}\right)^{\frac{1}{m}}. 
\end{align}
By substituting \eqref{lambda} into \eqref{pim1},   the optimal power allocation solution is given by 
 \begin{align}
 P_{m,i}  =&\left(\frac{e^R}{   \prod^{m}_{p=1} b_{m,p}|\bar{h}_{m,i}|^2}\right)^{\frac{1}{m}}  -\frac{|\bar{h}_{m,i}|^2\sum^{m-1}_{j=i}P_{j,i} +1}{  |\bar{h}_{m,i}|^2}.
\end{align}
The proof of the lemma is complete. 

\section{Proof for Lemma \ref{lemma2}}\label{proof2}

By using the assumption that the users are ordered according to their channel gains, i.e.,  $|h_m|^2> |h_{m+1}|^2$,   $|\bar{h}_{m,i}|^2$ can be simplified as follows:
\begin{align}
|\bar{h}_{m,i}|^2\triangleq \min \left\{
  | {h}_{i}|^2, \cdots,  | {h}_{m}|^2
 \right\}=|h_m|^2\triangleq\gamma_m,
 \end{align}
 which means that  
  the expression    of $P_{m,i}^{\rm H}$ can be simplified as follows: 
 \begin{align}\label{pmis1}
 P_{m,i}^{\rm H} = \left(\frac{e^R}{   \prod^{m}_{p=1}\frac{ \gamma_{m } }{ \gamma_{m } \sum^{m-1}_{j=p}P_{j,p}^{\rm H} +1}}\right)^{\frac{1}{m}}  -\frac{ \gamma_{m } \sum^{m-1}_{j=i}P_{j,i}^{\rm H} +1}{  \gamma_{m } }.
\end{align}
The lemma can be proved by mathematical induction. 

\subsection{The Base Case $m=2$ }
For the special case of $m=2$, by using \eqref{pmis1}, ${\rm U}_2$'s transmit power during the first two time slots are given by
  \begin{align} 
 P_{2,1}^{\rm H} =& \left(\frac{e^R}{   \frac{ \gamma_{2 }^2 }{ \gamma_{2 }  P_{1,1}^{\rm H} +1} }\right)^{\frac{1}{2}}  -\frac{  \gamma_{2 }  P_{1,1}^{\rm H} +1}{   \gamma_{2 } },\\\nonumber
 P_{2,2}^{\rm H} =& \left(\frac{e^R}{   \frac{ \gamma_{2 }^2 }{ \gamma_{2 }  P_{1,1}^{\rm H} +1} }\right)^{\frac{1}{2}}  -\frac{ 1}{    \gamma_{2 } }.
\end{align}
In addition, $P_{1,1}^{\rm H}= \frac{e^R-1}{\gamma_1}$ in order  to ensure $\log \left(1+ 
  \gamma_1 P_{1,1}^{\rm H}\right)=R$.

Therefore, $P_{2,1}^{\rm H}+P_{1,1}^{\rm H}$ is given by
\begin{align}
P_{2,1}^{\rm H}+P_{1,1}^{\rm H} = &\left(\frac{e^R}{   \frac{ \gamma_{2 }^2 }{ \gamma_{2 }  P_{1,1}^{\rm H} +1} }\right)^{\frac{1}{2}}  -\frac{  \gamma_{2 }  P_{1,1}^{\rm H} +1}{   \gamma_{2 } } +\frac{e^R-1}{\gamma_1}
\\\nonumber =&\left(\frac{e^R}{   \frac{ \gamma_{2 }^2 }{ \gamma_{2 }  P_{1,1}^{\rm H} +1} }\right)^{\frac{1}{2}}  -\frac{   1}{   \gamma_{2 } }  =P_{2,2}^{\rm H},
\end{align}
which means that the lemma holds for the base case.

 \subsection{Inductive Step}
 Assume that the lemma holds for the case of  $m$, which means that
 \begin{align}\label{assum}
 \sum^{m}_{i=1}P_{i,1}^{\rm H}=\cdots =  \sum^{m}_{i=m-1}P_{i,m-1}^{\rm H}=P_{m,m}^{\rm H}.
 \end{align}
 The aim of this section is to show that the lemma also holds for the case of $m+1$, i.e.,
   \begin{align}\label{aimpr1}
 \sum^{m+1}_{i=1}P_{i,1}^{\rm H}=\cdots =  \sum^{m+1}_{i=m}P_{i,m}^{\rm H}=P_{m+1,m+1}^{\rm H}.
 \end{align}
 
 Recall that ${\rm U}_{m+1}$'s transmit power in the first $(m+1)$ time slots can be written as follows:
  \begin{align}\label{pmis2}
 P_{m+1,i} ^{\rm H}=& \left(\frac{e^R}{   \prod^{m+1}_{p=1}\frac{ \gamma_{m+1 } }{ \gamma_{m+1 } \sum^{m}_{j=p}P_{j,p}^{\rm H} +1}}\right)^{\frac{1}{m+1}} \\\nonumber & -\frac{ \gamma_{m+1 } \sum^{m}_{j=i}P_{j,i} ^{\rm H}+1}{  \gamma_{m+1 } },
\end{align}
for $1\leq i \leq m+1$. 

We note that $ P_{m+1,i}^{\rm H} = P_{m+1,p}^{\rm H} $, for $1\leq i, p\leq m$, by using the assumption made in \eqref{assum}, as explained  in the following. The key observation is that the first term on the right-hand side of \eqref{pmis2} is the same for all $ P^{\rm H}_{m+1,i}$. Therefore, the conclusion that $ P^{\rm H}_{m+1,i} = P^{\rm H}_{m+1,p} $, for $1\leq i, p\leq m$, can be established if the following equality holds
  \begin{align} \label{newc3}
  \frac{ \gamma_{m+1 } \sum^{m}_{j=i}P^{\rm H}_{j,i} +1}{  \gamma_{m+1 } } =   \frac{ \gamma_{m+1 } \sum^{m}_{j=p}P^{\rm H}_{j,p} +1}{  \gamma_{m+1 } } ,
\end{align}
which is true given the assumption made in \eqref{assum}. Because $ P^{\rm H}_{m+1,i} = P^{\rm H}_{m+1,p} $, for $1\leq i, p\leq m$, the use of  \eqref{assum} leads to the following conclusion:
   \begin{align}
 \sum^{m+1}_{i=1}P^{\rm H}_{i,1}=\cdots =  \sum^{m+1}_{i=m}P^{\rm H}_{i,m},
 \end{align}
 which proves a part of \eqref{aimpr1}.  
 Therefore, the proof of the lemma can be completed  by showing that   $\sum^{m+1}_{i=1}P^{\rm H}_{i,1}=P^{\rm H}_{m+1,m+1}$, which is challenging to    prove directly. We note that  ${\rm U}_{m+1}$'s achievable data rates in the first $m$ time slots are identical, i.e.,
    \begin{align}
 \log\left(
 1+ \frac{\gamma_{m+1}P^{\rm H}_{m+1,i}}{\gamma_{m+1} \sum^{m}_{j=i}P^{\rm H}_{j,i}}\right)= \log\left(
 1+ \frac{\gamma_{m+1}P^{\rm H}_{m+1,p}}{\gamma_{m+1} \sum^{m}_{j=p}P^{\rm H}_{j,p}}\right),
 \end{align}
 since $\sum^{m}_{j=i}P^{\rm H}_{j,i}  =  \sum^{m}_{j=p}P^{\rm H}_{j,p}$ as shown in    \eqref{assum} and   $ P^{\rm H}_{m+1,i} = P^{\rm H}_{m+1,p} $, for $1\leq i, p\leq m$. Therefore, the
 $ P^{\rm H}_{m+1,i}$, $1\leq i \leq m+1$, are also the optimal solution of the following optimization problem:
   \begin{problem}\label{pb:a1} 
  \begin{alignat}{2}
\hspace{-1em}\underset{P _{m+1,1},P _{m+1,m+1}   }{\rm{min}} &\quad   mP _{m+1,1}+ P _{m+1,m+1}  \label{a1tst:1}
\\ \nonumber s.t. &\quad  m \log\left(
 1+ \frac{\gamma_{m+1}P _{m+1,1}}{\gamma_{m+1} \sum^{m}_{j=1}P _{j,1}+1}\right)\\
 &\quad +  \log\left(
 1+  \gamma_{m+1}P _{m+1,m+1}\right) \geq R.    \label{a1tst:2} 
  \end{alignat}
\end{problem} 
The difference between problems \eqref{pb:a1} and \eqref{pb:2} is that  in problems \eqref{pb:a1}, the first $m$ time slots are merged together to become a single time slot with duration $mT$. In other words,  problem \eqref{pb:a1} can be viewed as a simple two-user case. 
By following   steps similar to those for solving problem \eqref{pb:2},  the alternative expressions     for $P^{\rm H}_{m+1,1}$ and $P^{\rm H}_{m+1,m+1}$ can be obtained as follows:
\begin{align}
P^{\rm H}_{m+1,1} =& \left(\frac{e^R}{\psi^m\gamma_{m+1}}\right)^{\frac{1}{m+1}}
-\frac{I_m\gamma_{m+1}+1}{\gamma_{m+1}},\\\nonumber
 P^{\rm H}_{m+1,m+1} = & \left(\frac{e^R}{\psi^m\gamma_{m+1}}\right)^{\frac{1}{m+1}} - \frac{1}{\gamma_{m+1}},
\end{align} 
where $I_m = \sum^{m}_{j=1}P^{\rm H}_{j,1}$ and $\psi=\frac{\gamma_{m+1}}{I_m\gamma_{m+1}+1}$. 

Proving $\sum^{m+1}_{i=1}P^{\rm H}_{i,1}=P^{\rm H}_{m+1,m+1}$ is equivalent to showing $P^{\rm H}_{m+1,1} +I_m=P^{\rm H}_{m+1,m+1}$, which holds since
\begin{align}
P^{\rm H}_{m+1,1} +I_m=& \left(\frac{e^R}{\psi^m\gamma_{m+1}}\right)^{\frac{1}{m+1}}
-\frac{I_m\gamma_{m+1}+1}{\gamma_{m+1}}+ I_m\\\nonumber
=  & \left(\frac{e^R}{\psi^m\gamma_{m+1}}\right)^{\frac{1}{m+1}} - \frac{1}{\gamma_{m+1}}= P^{\rm H}_{m+1,m+1} .
\end{align} 
Therefore,   the proof of the lemma is complete. 

\section{Proof for Lemma \ref{lemma3}}\label{proof3}
Recall that the following  observations:
\begin{enumerate}
\item The power allocation solution shown in Lemma \ref{lemma1} is an optimal solution of problem \eqref{pb:2};
\item The feasible set of problem \eqref{pb:2} is larger than that of problem \eqref{pb:1.5}, which means that the optimal value of problem \eqref{pb:2} is no larger than that of problem \eqref{pb:1.5}; 
\end{enumerate}
By using these facts, the lemma that the power allocation solution shown in Lemma \ref{lemma1} is the optimal solution of problem \eqref{pb:1.5} can be proved, by showing that  this solution is feasible to problem \eqref{pb:1.5}, i.e.,   $P^{\rm H}_{m,i}>0$, $1\leq i \leq m$.

We first note that ${\rm U}_m$ always     uses the $m$-th time slot (the OMA time slot), i.e.,   $P^{\rm H}_{m,m}>0$, which can be straightforwardly established, as shown in the following. Recall from \eqref{pmis1} that $ P^{\rm H}_{m,m}$ can be expressed as follows:
 \begin{align} 
 P^{\rm H}_{m,m} = \left(\frac{e^R}{   \prod^{m}_{p=1}\frac{ \gamma_{m } }{ \gamma_{m } \sum^{m-1}_{j=p}P^{\rm H}_{j,p} +1}}\right)^{\frac{1}{m}}  -\frac{ 1}{  \gamma_{m } }.
\end{align}
If $ P^{\rm H}_{m,m}\leq 0$, the following inequality needs to hold
 \begin{align} 
 \left(\frac{e^R}{   \prod^{m}_{p=1}\frac{ \gamma_{m } }{ \gamma_{m } \sum^{m-1}_{j=p}P^{\rm H}_{j,p} +1}}\right)^{\frac{1}{m}} \leq  \frac{ 1}{  \gamma_{m } }.
\end{align}
The above inequality is equivalent to the following: 
 \begin{align} \label{xddd1}
 e^R   \leq     \prod^{m}_{p=1}\frac{ 1 }{ \gamma_{m } \sum^{m-1}_{j=p}P^{\rm H}_{j,p} +1}.
\end{align}
We note that $e^R>1$, since $R>0$. However,   $\prod^{m}_{p=1}\frac{ 1 }{ \gamma_{m } \sum^{m-1}_{j=p}P^{\rm H}_{j,p} +1}\leq 1$, which means the inequality in \eqref{xddd1} cannot hold, and hence $ P^{\rm H}_{m,m} $ is strictly    positive. 

Therefore, the lemma can be proved showing   that ${\rm U}_m$'s hybrid NOMA power allocations   for  the first $m-1$ time slots need to be  also strictly positive, i.e., $P^{\rm H}_{m,i}>0$, for $1\leq i \leq m-1$, which   can be proved by mathematical induction. 

\subsection{The Base Case $m=2$ }
For the special case of $m=2$, by using \eqref{pmis1}, ${\rm U}_2$'s transmit power during the first time slot is given by
 \begin{align}
P^{\rm H}_{2,1} =& \left(\frac{e^R}{   \frac{ \gamma_{2 }^2 }{ \gamma_{2 }  P^{\rm H}_{1,1} +1} }\right)^{\frac{1}{2}}  -\frac{  \gamma_{2 }  P^{\rm H}_{1,1} +1}{   \gamma_{2 } }.
\end{align}
Therefore,  the OMA mode is used if $ P^{\rm H}_{2,1}\leq0$, i.e., 
  \begin{align}\label{oma2}
\left(\frac{e^R}{   \frac{ \gamma_{2 }^2 }{ \gamma_{2 }  P^{\rm H}_{1,1} +1} }\right)^{\frac{1}{2}}  \leq \frac{  \gamma_{2 }  P^{\rm H}_{1,1} +1}{   \gamma_{2 } },
\end{align}
 which can be simplified as follows: 
   \begin{align}
R    \leq   \log \left(1+\gamma_2P^{\rm H}_{1,1}\right) < \log \left(1+ 
  \gamma_1 P^{\rm H}_{1,1}\right),
\end{align}
where the last inequality follows by the fact that $\gamma_1>\gamma_2$. 
By using the fact that   $P^{\rm H}_{1,1}$ is chosen to ensure $\log \left(1+ 
  \gamma_1 P^{\rm H}_{1,1}\right)=R$, the following contradiction can be established:
     \begin{align}
R      < \log \left(1+ 
  \gamma_1 P^{\rm H}_{1,1}\right)=R,
\end{align}
which means that $P_{2,1}>0$, and hence for the special case of $m=2$, the lemma holds. 

 \subsection{Inductive Step}
 Assume that the lemma holds for the case of $m$, i.e.,    $P^{\rm H}_{j,i}>0$,  $i\leq m-1$ and $j\leq m$, which makes    Lemma \ref{lemma2}   applicable and  leads to  the following equality:
 \begin{align}\label{lemmaxx}
 \sum^{m}_{j=1}P^{\rm H}_{j,1}=\cdots =  \sum^{m}_{j=m-1}P^{\rm H}_{j,m-1}=P^{\rm H}_{m,m}\triangleq I_m.
 \end{align}

 The aim of this section is to prove that the lemma also holds for the case of $m+1$, i.e., $P^{\rm H}_{m+1,i}>0$, $i\leq m$, which is also challenging to   prove directly.      Recall that ${\rm U}_{m+1}$'s achievable data rates   during the first $m$ time slots are given by
 \begin{align}
 R_{m+1,i} = \log\left(
 1+\frac{\gamma_{m+1}P^{\rm H}_{m+1,i}}{\gamma_{m+1}\sum^{m}_{j=i}P^{\rm H}_{j,i}+1 }
 \right),
 \end{align}
 for $1\leq i \leq m$. The use of \eqref{lemmaxx} indicates that ${\rm U}_{m+1}$ suffers the same amount of interference ($I_m$) during each of the first $m$ time slots. In addition, the user's channel gains during the first $m$ time slots are  also the same. Therefore,   ${\rm U}_{m+1}$'s transmit powers during the first $m$ time slots, i.e., $P^{\rm H}_{m+1,i}$, $1\leq i \leq m$, must be the same, i.e.,  $P^{\rm H}_{m+1,1}=\cdots =P^{\rm H}_{m+1,m}$. In this case,  ${\rm U}_{m+1}$'s transmit powers can    be alternatively obtained from the following optimization problem:
    \begin{problem}\label{pb:a2} 
  \begin{alignat}{2}
\hspace{-1em}\underset{P_{m+1,1}, P_{m+1,m+1}     }{\rm{min}} &\quad   mP_{m+1,1}+ P_{m+1,m+1}  \label{a2tst:1}
\\ \nonumber s.t. &\quad  m \log\left(
 1+ \frac{\gamma_{m+1}P_{m+1,1}}{\gamma_{m+1}I_m+1}\right)\\
 &\quad +  \log\left(
 1+  \gamma_{m+1}P_{m+1,m+1}\right) \geq R,    \label{a2tst:2} 
  \end{alignat}
\end{problem}
which is identical to problem \eqref{pb:a1}. Therefore,  $P^{\rm H}_{m+1,1}$ can be expressed as follows:
\begin{align}
P^{\rm H}_{m+1,1} =& \left(\frac{e^R}{\psi^m\gamma_{m+1}}\right)^{\frac{1}{m+1}}
-\frac{I_m\gamma_{m+1}+1}{\gamma_{m+1}} .
\end{align} 
To ensure $P^{\rm H}_{m+1,1} >0$, the following inequality needs to hold:
\begin{align}
  \left(\frac{e^R}{\psi^m\gamma_{m+1}}\right)^{\frac{1}{m+1}}
>\frac{I_m\gamma_{m+1}+1}{\gamma_{m+1}} ,
\end{align} 
which can be rewritten as follows
\begin{align}\label{ineq1}
 R 
>&\log\left(\frac{\psi^m \left(I_m\gamma_{m+1}+1\right)^{m+1}}{h^{m}_{m+1}} \right).
\end{align} 
Recall that $\psi=\frac{\gamma_{m+1}}{I_m\gamma_{m+1}+1}$, and hence the inequality in \eqref{ineq1} can be expressed as follows:
\begin{align}\nonumber
 R 
>&\log\left(\frac{ \frac{h^m_{m+1}}{(I_m \gamma_{m+1}+1)^m} \left(I_m\gamma_{m+1}+1\right)^{m+1}}{h^{m}_{m+1}} \right)\\ 
 =&\log\left(  I_m\gamma_{m+1}+1 \right).\label{ineq2}
\end{align} 
In order to show that $ R 
> \log\left(  I_m\gamma_{m+1}+1 \right)$, the fact that the users' channel gains are ordered can be used to show  the following inequality: 
\begin{align}\label{cdsx1}
\log\left(  I_m\gamma_{m+1}+1 \right)<\log\left(  I_m\gamma_{m}+1 \right)=\log\left(  P^{\rm H}_{m,m}\gamma_{m}+1 \right),
\end{align}
where the last step follows from the  equality in \eqref{lemmaxx}.  

The   stationarity of the KKT conditions for problem \eqref{pb:2} leads to the following conclusion:
 \begin{align}\label{cfde}
&\sum^{m-1}_{i=1} \log\left(
 1+\frac{\gamma_{m}P^{\rm H}_{m,i}}{\gamma_{m}\sum^{m-1}_{j=i}P^{\rm H}_{j,i}+1 }
 \right)\\\nonumber &+\log\left(  P^{\rm H}_{m,m}\gamma_{m}+1 \right)=R.
 \end{align}
By using \eqref{cfde} and also   the assumption that $P^{\rm H}_{m,i}>0$ for $1\leq i\leq m-1$, the following inequality can be established: 
 \begin{align}\label{cfde1}
 \log\left(  P^{\rm H}_{m,m}\gamma_{m}+1 \right)<R.
 \end{align}
By combining \eqref{cdsx1} with \eqref{cfde1}, $R>\log\left(  I_m\gamma_{m+1}+1 \right)$ is proved, which means that $P^{\rm H}_{m+1,i}>0$, $1\leq i\leq m$.  Therefore,   the proof of the lemma is complete.

\section{Proof of Lemma \ref{lemma4}}\label{proof4x}
The lemma can be again proved by mathematical induction. For the base case $m=1$,   it is straightforward to show that there is a single optimal solution for problem \eqref{pb:1.5}. 

For the inductive step, assume that for the case $m-1$, the lemma holds, i.e., ${\rm U}_i$, $1\leq i \leq m-1$, chooses hybrid NOMA,  which makes   Lemma \ref{lemma2} applicable.  The aim of the proof is to show that  the lemma holds for the case of $m$.
 
Since problem \eqref{pb:1.5} is convex, its optimal solution needs to satisfy the KKT conditions. By analyzing the KKT conditions,  it is straightforward to show that if there exists another optimal solution for problem \eqref{pb:1.5}, one or multiple $P_{m,i}$ need to be  zero.     Without loss of generality, denote $\mathcal{S}$ as the subset that collects the indices of $P_{m,i}$,   which are   zero, i.e.,   $P_{m,i}= 0$, for $i\in \mathcal{S}$. Define $\mathcal{S}^c$ as the complementary set of $\mathcal{S}$. We note that $m$ must be  included in $\mathcal{S}^c$ since  $P_{m,m}$ cannot be zero, as discussed in the proof of Lemma \ref{lemma3}.  Therefore,     the use of the KKT condition shown in \eqref{kkt1x} leads to the following conclusion:  
\begin{align}\label{kkty1x}
 1 - \lambda  b_{m,i}\gamma_m =0,  i\in\mathcal{S},
\end{align} 
which can be rewritten as follows:
\begin{align}\label{kkty1x2}
 0=1 - \frac{\lambda   \gamma_m}{\gamma_m\sum^{m-1}_{j=i}P_{j,i} +1} =1 - \frac{\lambda   \gamma_m}{\gamma_mP_{m-1,m-1} +1} ,  
\end{align} 
where the first step follows by $b_{m,i}=\frac{1}{\gamma_m\sum^{m-1}_{j=i}P_{j,i} +1}$, and the last step follows by Lemma \ref{lemma2}. Therefore, the fact that  $P_{m,i}=0$,   $i\in \mathcal{S}$, leads to the following expression of   $\lambda$:
\begin{align}\label{kkty1x2}
 \lambda = P_{m-1,m-1} +\frac{1}{\gamma_m}.
\end{align} 

On the other hand,   \eqref{kkt1x} can be used to obtain  the following expression for $ P_{m,j} $,  $j\in \mathcal{S}^c$:
 \begin{align}\label{pim1xdd}
 P_{m,j}  =\lambda  - \frac{1}{b_{m,i}\gamma_m}, j\in \mathcal{S}^c.
\end{align}
Because of the complementary slackness condition, the  $P_{m,j}$, $j\in \mathcal{S}^c$, need to satisfy the following condition:
\begin{align}
\sum_{j\in\mathcal{S}^c}\log\left(1+b_{m,j}\gamma_{m}P_{m,j}\right) = R,
\end{align}
which can be rewritten  as follows: 
\begin{align}
\prod_{j\in\mathcal{S}^c}\lambda b_{m,j}\gamma_m  =e^ R.
\end{align}
By using the expressions of $b_{m,j}$ and $\lambda$ in \eqref{kkty1x2}, the above equality can be expressed as follows:
\begin{align}\label{dfde1}
\left(\gamma_mP_{m-1,m-1} +1\right)^{ |\mathcal{S}^c|}\prod_{j\in\mathcal{S}^c}  \frac{1}{\gamma_m\sum^{m-1}_{j=i}P_{j,i} +1} =e^ R.
\end{align}
By applying Lemma \ref{lemma2}, \eqref{dfde1} can be simplified as follows:
\begin{align}\label{df32}
 \gamma_mP_{m-1,m-1} +1 =e^ R.
\end{align}
Recall that $P_{m-1,m-1}$ is a function of $\gamma_{m-1}$ and not related to $\gamma_m$. Therefore, the probability for the equality in \eqref{df32} to hold is zero since    the users' channel gains are independent fading.   Therefore, the solution shown in Lemma \ref{lemma1} is the only optimal solution of problem \eqref{pb:1.5}, which completes the proof of Lemma \ref{lemma4}. 
\bibliographystyle{IEEEtran}
\bibliography{IEEEfull,trasfer}
  \end{document}